\documentclass[accepted=2025-08-09,a4paper,onecolumn,11pt,arxiv]{quantumarticle}
\pdfoutput=1
\usepackage[utf8]{inputenc}
\usepackage[english]{babel}
\usepackage[T1]{fontenc}
\usepackage{amsmath}
\usepackage{amssymb}
\usepackage{graphicx}
\usepackage{amsthm}
\usepackage{hyperref}
\usepackage{bbold}
\usepackage{xcolor}
\usepackage{longtable}
\usepackage{tcolorbox}
\usepackage{booktabs}
\usepackage{adjustbox}
\usepackage{algorithm}
\usepackage{array}
\usepackage{subcaption}
\usepackage{caption}
\usepackage{lipsum}
\usepackage{enumitem}
\usepackage[capitalize,nameinlink]{cleveref}

\usepackage{tikz}
\usepackage{lipsum}

\newtheorem{proposition}{Proposition}
\newtheorem{theorem}{Theorem}
\newtheorem{lemma}{Lemma}
\newtheorem{corollary}{Corollary}

\numberwithin{proposition}{section}
\numberwithin{theorem}{section}
\numberwithin{lemma}{section}
\numberwithin{corollary}{section}
\numberwithin{remark}{section}
\numberwithin{definition}{section}
\numberwithin{equation}{section}

\begin{document}

\newcommand*{\cl}[1]{{\mathcal{#1}}}
\newcommand*{\bb}[1]{{\mathbb{#1}}}
\newcommand{\ket}[1]{|#1\rangle}
\newcommand{\bra}[1]{\langle#1|}
\newcommand{\inn}[2]{\langle#1|#2\rangle}
\newcommand{\proj}[2]{| #1 \rangle\!\langle #2 |}
\newcommand*{\tn}[1]{{\textnormal{#1}}}
\newcommand*{\1}{{\mathbb{1}}}
\newcommand{\T}{\mbox{$\textnormal{Tr}$}}
\newcommand{\todo}[1]{\textcolor[rgb]{0.99,0.1,0.3}{#1}}
\newcommand{\norm}[1]{\left\lVert#1\right\rVert}
\newcommand{\abs}[1]{\left\lvert#1\right\rvert}
\newcommand*\rfrac[2]{{}^{#1}\!/_{#2}}%running fraction with slash - requires math mode.

\title{Resource-efficient algorithm for estimating the trace of quantum state powers}

\author{Myeongjin Shin$^*$}
\email{hanwoolmj@kaist.ac.kr}
\homepage{https://myeongjinshin.github.io/}
\affiliation{School of Computing, KAIST, Daejeon 34141, Korea}
\affiliation{Team QST, Seoul National University, Seoul 08826, Korea}

\author{Junseo Lee$^*$}
\email{harris.junseo@gmail.com}
\homepage{https://harris-junseo-lee.github.io/}
\affiliation{Team QST, Seoul National University, Seoul 08826, Korea}
\affiliation{Quantum AI Team, Norma Inc., Seoul 04799, Korea}

\author{Seungwoo Lee}
\email{smilelee9@kaist.ac.kr}
\affiliation{School of Computing, KAIST, Daejeon 34141, Korea}
\affiliation{Team QST, Seoul National University, Seoul 08826, Korea}

\author{Kabgyun Jeong}
\email{kgjeong6@snu.ac.kr}
\affiliation{Team QST, Seoul National University, Seoul 08826, Korea}
\affiliation{Research Institute of Mathematics, Seoul National University, Seoul 08826, Korea}
\affiliation{School of Computational Sciences, Korea Institute for Advanced Study, Seoul 02455, Korea}
\thanks{\\\mbox{} \\ $^*$The first two authors contributed equally to this work.}

\maketitle
\begin{abstract}
Estimating the trace of quantum state powers, $\T(\rho^k)$, for $k$ identical quantum states is a fundamental task with numerous applications in quantum information processing, including nonlinear function estimation of quantum states and entanglement detection. On near-term quantum devices, reducing the required quantum circuit depth, the number of multi-qubit quantum operations, and the copies of the quantum state needed for such computations is crucial. In this work, inspired by the Newton-Girard method, we significantly improve upon existing results by introducing an algorithm that requires only $\mathcal{O}(\widetilde{r})$ qubits and $\mathcal{O}(\widetilde{r})$ multi-qubit gates, where $\widetilde{r} = \min\left\{\textnormal{rank}(\rho), \left\lceil\ln\left({2k}/{\epsilon}\right)\right\rceil\right\}$. This approach is efficient, as it employs the $\tilde{r}$-entangled copy measurement instead of the conventional $k$-entangled copy measurement, while asymptotically preserving the known sample complexity upper bound. Furthermore, we prove that estimating $\{\T(\rho^i)\}_{i=1}^{\tilde{r}}$ is sufficient to approximate $\T(\rho^k)$ even for large integers $k > \widetilde{r}$. This leads to a rank-dependent complexity for solving the problem, providing an efficient algorithm for low-rank quantum states while also improving existing methods when the rank is unknown or when the state is not low-rank. Building upon these advantages, we extend our algorithm to the estimation of $\T(M\rho^k)$ for arbitrary observables and $\T(\rho^k \sigma^l)$ for multiple quantum states.
\end{abstract}

\newpage
\tableofcontents

\newpage
\section{Introduction}\label{chap:intro}

\subsection{Trace of quantum state powers}\label{chap:top}
Estimation task for the trace of the product of identical density matrices, which is represented as
\begin{equation*}
    \T(\rho^k) \quad \textit{`trace of quantum state powers'}
\end{equation*}
given access to copies of the quantum state $\rho$, is a core subroutine for many algorithms and applications in quantum information theory. We refer to this quantity as the `trace of quantum state powers,' which is used to calculate the value of integer Rényi entropy~\cite{johri2017entanglement, elben2018renyi, vermersch2018unitary}, nonlinear functions of quantum states~\cite{ekert2002direct, brun2004measuring, van2012measuring, zhou2024hybrid, bovino2005direct}, and deducing the eigenvalues of the quantum state, a process known as entanglement spectroscopy~\cite{johri2017entanglement, yirka2021qubit}.

We focus on estimating $\T(\rho^k)$ for large integer $k$. The main applications are calculating the nonlinear functions of quantum states, which need estimation of the trace of large powers. Precisely, Yirka and Suba{\c{s}}{\i}~\cite{yirka2021qubit} proved that the trace of `well-behaved’ polynomials $g(\rho)$, such as $g(x) = (1+x)^{\alpha}$ and $\log(1+x)$, can be efficiently estimated using the trace of quantum state powers. Moreover, $\T(e^{\beta\rho})$ is an example with applications in thermodynamics.

The preparation of quantum Gibbs states~\cite{wang2021variational, consiglio2024variational, terhal2000problem, riera2012thermalization, wu2019variational} is an essential part of quantum computation, used in various applications such as quantum simulation, quantum optimization, and quantum machine learning. The truncated Taylor series
\begin{equation}
S_k(\rho) = \sum_{i=1}^k \T((\rho-I)^i\rho)
\end{equation}
is exploited as the cost function for variational quantum Gibbs state preparation~\cite{wang2021variational}, which can be calculated by $\{\T(\rho^i)\}_{i=1}^{k+1}$.

Several methods for the estimation of the trace of quantum state powers have been proposed, such as the generalized swap test~\cite{ekert2002direct}, entanglement spectroscopy via Hadamard test~\cite{johri2017entanglement}, two-copy test~\cite{subacsi2019entanglement}, qubit-efficient entanglement spectroscopy~\cite{yirka2021qubit}, multivariate trace estimation~\cite{quek2024multivariate}, and methods using randomized measurement such as classical shadows~\cite{van2012measuring, huang2020predicting, rath2021quantum}. An analysis of these methods is performed in~\cref{chap:review}.

Our work is inspired by the Newton-Girard method, as demonstrated in~\cref{chap:NG}. Specifically, we use quantum devices only to estimate $\{\T(\rho^i)\}_{i=1}^{\Tilde{r}}$, where $\widetilde{r}=\min\left\{r, \left\lceil\ln\left({2k}/{\epsilon}\right)\right\rceil\right\}$. (From this point onward, we consistently use $r$ to denote the rank of the quantum state $\rho$ throughout the paper.) Subsequently, we use a classical computer with a recursive formula to calculate $\{\T(\rho^i)\}_{i=1}^k$, with an additive error of less than $\epsilon$ for large $k\in\mathbb{N}$. In~\cref{chap:rank}, we prove that the rank $r$ is sufficient, implying that quantum devices are required only for $\{\T(\rho^i)\}_{i=1}^r$. By defining the notion of `effective rank' in~\cref{chap:effective}, we further prove a more advanced theorem that the effective rank $\widetilde{r}$ is sufficient for estimating the trace of quantum state powers. The Newton-Girard method and recursion are used in the proof. 

Furthermore, we argue that combining our work with previous ones~\cite{johri2017entanglement, ekert2002direct, yirka2021qubit, subacsi2019entanglement, quek2024multivariate, buhrman2001quantum, gottesman2001quantum} improves its algorithmic performance. The number of needed qubits (i.e., width of the circuit) and the required multi-qubit gates are reduced. We support our work with numerical simulations. To emphasize the importance of our work, we demonstrate advantages when applying our method to applications such as calculating nonlinear functions of quantum states, preparation of quantum Gibbs states, and entanglement detection.

\begin{table*}[ht]
\centering
\resizebox{\textwidth}{!}{
\centering
\begin{tabular}{c|c|c}
\hline
\textbf{Quantity} & \textbf{Quantum Resource Needed} & \textbf{Upper bound on $t$} \\ \hline\hline
\begin{tabular}[c]{@{}c@{}}$\T(\rho^k)$\\ \small{\cref{thm:rank},~\ref{thm:eff}}\end{tabular} & $\{\T(\rho^i)\}_{i=1}^t$ & $\min\left\{\text{rank}(\rho), \left\lceil\ln\left({2k}/{\epsilon}\right)\right\rceil\right\}$ \\ \hline
\begin{tabular}[c]{@{}c@{}}$\T(M\rho^k)$\\ \small{\cref{thm:obs}}\end{tabular} & $ \{\T(\rho^i)\}_{i=1}^t $, $ \{\T(M\rho^i)\}_{i=1}^t $ & $\min\left\{\text{rank}(\rho), \left\lceil\ln\left({2k\norm{M}_\infty}/{\epsilon}\right)\right\rceil\right\}$ \\ \hline
\begin{tabular}[c]{@{}c@{}}$\T(\rho^k \sigma^l)$\\ \small{\cref{thm:multi}}\end{tabular} & $ \{\T(\rho^i)\}_{i=1}^t $, $ \{\T(\sigma^i)\}_{i=1}^t $, $\{\T(\rho^i\sigma^j)\}_{(i,j)=(1,1)}^{(t,t)}$ & $\min\left\{\max\{\textnormal{rank}(\rho),\textnormal{rank}(\sigma)\}, \left\lceil\ln\left({(4k + 4l)}/{\epsilon}\right)\right\rceil\right\}$ \\ \hline
\end{tabular}}
\caption{{\textbf{Summary of quantum resource requirements and effective rank conditions for $\epsilon$-additive estimations.} This table summarizes the key results of the paper. It presents the range of values to be obtained through quantum resources for each of the three physical quantities, and further details can be found in the corresponding theorems.}}
\label{tab:main}
\end{table*}

\subsection{Organization of the paper}
Our paper is structured as follows. In~\cref{chap:review}, we review existing results on attempts to estimate the trace of quantum state powers. This includes results derived from variations of the swap test, several other approaches, and key related studies. In~\cref{chap:NG}, we introduce the Newton-Girard method, which serves as the fundamental principle of our algorithm. Then, in~\cref{chap:algo}, we describe how we specifically design our algorithm using this method. Subsequently, we analyze our algorithm in detail. In~\cref{chap:rank}, we explain how the quantum resources required for our algorithm are related to the rank of the quantum state. In~\cref{chap:effective}, we strengthen our algorithm by introducing the concept of effective rank, allowing it to be applied even when the exact rank of the quantum state is unknown. In~\cref{chap:obs}, we extend the problem to the case of arbitrary observables, which is a more generalized version based on the quantum resources required for estimating the trace of quantum state powers, as determined in previous sections. \cref{chap:simul} discusses the results of numerical simulations demonstrating the operation of our algorithm, and~\cref{chap:application} explores how our algorithm can be applied to other quantum information tasks. Finally, in~\cref{chap:con}, we summarize our study, discuss its limitations, and outline potential directions for future research. The proofs of all the theorems, corollaries, and lemmas presented in the paper are provided in~\cref{chap:proofs}. The table summarizing our main results is presented in~\cref{tab:main}. The $\widetilde{\mathcal{O}}(\cdot)$ asymptotic notation used in our paper hides polylogarithmic factors in certain variables. To ensure clarity in each context, we explicitly specify which variables are involved whenever necessary, and repeat the explanation when appropriate. Throughout this paper, unless otherwise noted, $n$ denotes the number of qubits and $d$ denotes the dimension of the quantum state, with $d = 2^n$.

\subsection{Literature review}\label{chap:review}
The swap test (ST)~\cite{buhrman2001quantum, gottesman2001quantum, fanizza2020beyond, foulds2021controlled, gitiaux2022swap} estimates $\T(\rho\sigma)$, the trace of the product of two matrices $\rho$ and $\sigma$:
\begin{equation}\label{swap}
    \T\left(S\left(\rho\otimes\sigma\right)\right) = \T\left(\rho\sigma\right),
\end{equation}
where $S$ denotes the swap operator. The ST can be performed using 1 ancilla qubit with 1 controlled-SWAP (\textsf{CSWAP}) operation and 2 Hadamard gates as shown in~\cref{fig:swapcirc}. The ST can be thought of as performing the observable $S$ on~\cref{swap}. The observation that quantities like $\T(\rho\sigma)$ can be estimated without the need for full-state tomography was a significant development.

Following this line of thinking, Ekert \textit{et al.}~\cite{ekert2002direct} proposed a cyclic shift permutation operator $W^{\pi}$ for a generalized ST:
\begin{equation}\label{gswap}
    \T\left(W^{\pi}\left(\rho_1\otimes\ldots\otimes\rho_k\right)\right) = \T\left(\rho_1\ldots\rho_k\right).
\end{equation}
By using~\cref{gswap} above, the trace of quantum state powers $\T(\rho^k)$ can be easily calculated. Note that regardless of the dimension $d=\text{dim}(\rho)$ and the number of quantum states $k$, the generalized ST needs only $\mathcal{O}\left({{1}/{\epsilon^{2}}}\right)$ runs on a quantum device for $\epsilon$ additive error prediction. Thus, $\mathcal{O}\left({k}/{\epsilon^{2}}\right)$ copies are needed for the estimation of $\T(\rho^k)$. This method requires $\mathcal{O}(k)$ qubits, a quantum circuit of $\mathcal{O}(k)$ depth, and $\mathcal{O}(k)$ multi-qubit gates.

Various methods have been proposed for better estimation~\cite{johri2017entanglement, yirka2021qubit, subacsi2019entanglement, quek2024multivariate} of the trace of quantum state powers. A comparison of these methods is shown in~\cref{tab:compare}.

\begin{figure}[t]
  \centering
  \includegraphics[width=6cm]{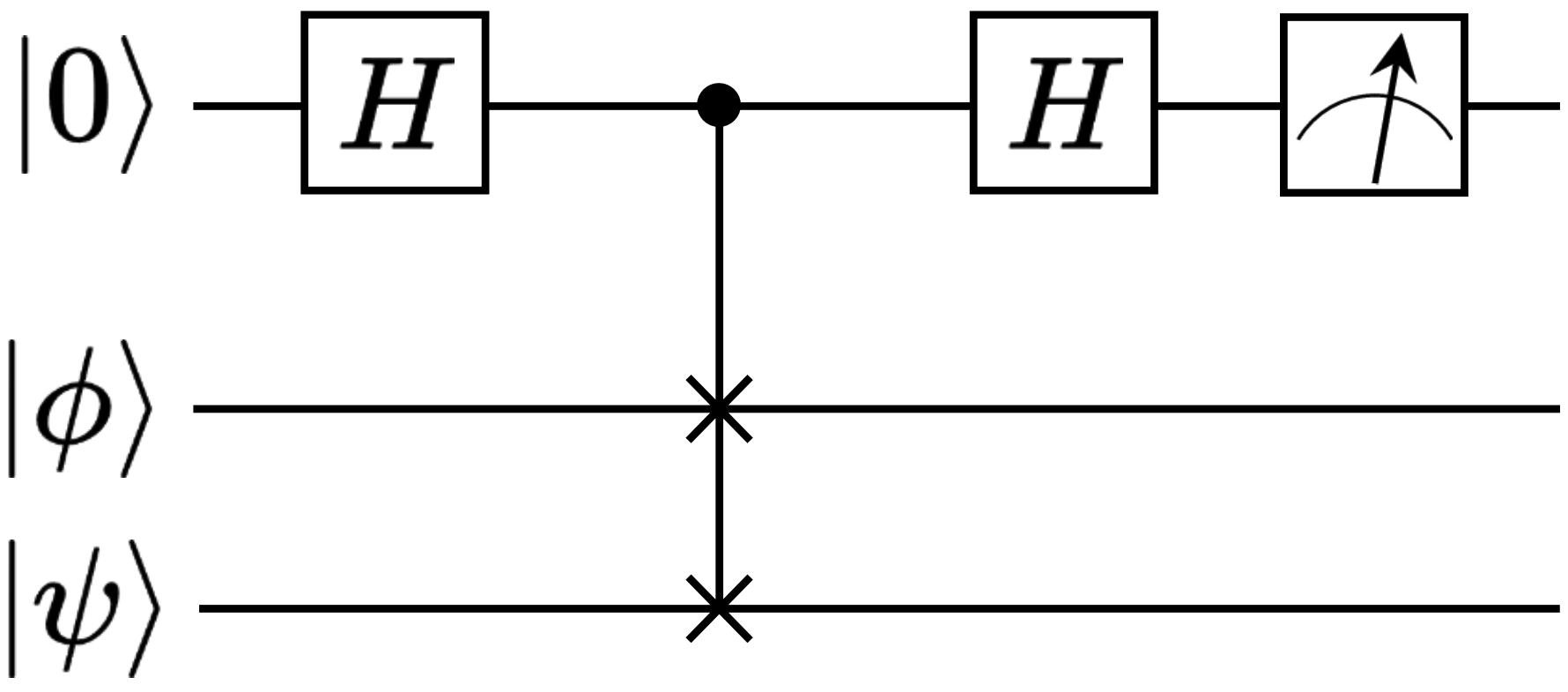}
  \caption{\textbf{Circuit implementing the swap test between two states.} The simplest case of a quantum circuit for calculating the trace of the product of two density matrices using the swap test is illustrated. It shows that 2 single-qubit gates, 1 three-qubit gate, and 1 ancilla qubit are required.}
  \label{fig:swapcirc}
\end{figure}

The entanglement spectroscopy via hadamard test (HT)~\cite{johri2017entanglement} is a generalized algorithm that estimates the expectation value of an arbitrary unitary operator or observable $M$. Specifically, ST can be thought of as a special case of HT when $M=S$. The HT has linear depth $\mathcal{O}(k)$ and uses $\mathcal{O}(k)$ copies of the state. A more improved algorithm, the entanglement spectroscopy via two-copy test (TCT)~\cite{subacsi2019entanglement}, achieves constant depth and uses $\mathcal{O}(k)$ copies of the state. Thus, both use $\mathcal{O}(k)$ qubits in the estimation circuit. That is, both HT and TCT are improved algorithms but need the original entangled pure state $\ket{\psi}_{AB}$ for the estimation of $\T(\rho_A^k)$, where $\rho_A = \T_B(\ket{\psi}\bra{\psi}_{AB})$.

Qubit-efficient entanglement spectroscopy~\cite{yirka2021qubit} employs qubit-reset strategies to reduce the number of qubits in the quantum circuit. This method requires only $n$ qubits, constant in terms of power $k$. When combined with TCT, it requires a linear circuit depth $\mathcal{O}(k)$. Also, Yirka and Suba{\c{s}} \cite{yirka2021qubit} defines the notion of `effective depth,' and TCT with qubit-reset strategy requires only a constant effective circuit depth $\mathcal{O}(1)$. However, this qubit-reset strategy still demands $\mathcal{O}(k)$ copies of the original entangled pure state $\ket{\psi}_{AB}$, and qubit-reset could lead to more vulnerability to noise.

Without the need for the entangled pure state $\ket{\psi}_{AB}$, multivariate trace estimation~\cite{oszmaniec2024measuring, chien2016characterization, bargmann1964note} $\T(\rho_1\rho_2\ldots\rho_k)$, a general case of the trace of quantum state powers, has been proposed with constant quantum depth~\cite{quek2024multivariate}. Inspired by the method of Shor error correction~\cite{shor1996fault}, this approach requires only constant quantum circuit depth, utilizing $\mathcal{O}(k)$ multi-qubit gates and $\mathcal{O}(k)$ qubits, and establishes numerous applications for multivariate trace and trace-of-powers estimation. By combining our work with these advancements, we provide an advantageous solution for estimating $\T(\rho^k)$ with large $k$. Specifically, leveraging multivariate trace estimation~\cite{quek2024multivariate}, we can reduce the number of required qubits from $\mathcal{O}(k)$ to $\mathcal{O}(\widetilde{r})$ and multi-qubit gates from $\mathcal{O}(k)$ to $\mathcal{O}(\widetilde{r})$ for $\T(\rho^k)$ estimation, where $\widetilde{r} = \min\left\{r, \left\lceil\ln\left({2k}/{\epsilon}\right)\right\rceil\right\}$.

There are alternative methods that use classical shadows~\cite{huang2020predicting, rath2021quantum} to estimate $\T(\rho^k)$. Using 
\begin{equation}
    \T\left(W^{\pi}\left(\rho\otimes\ldots\otimes\rho\right)\right) = \T(\rho^k),
\end{equation}
and linearly combining the classical snapshots of $\rho$, we can obtain a classical random variable whose expectation is $\T(\rho^k)$. The advantage of these alternative methods is that they allow for measurements to be taken sequentially and do not rely on the assumption that the samples of $\rho$ used by the algorithm are identical and independent~\cite{quek2024multivariate}. However, due to the exponential scaling of $\T ((W^{\pi})^2)$, the sample and computational complexity are exponential in terms of qubits. So, this method requires the number of copies as the dimension $d$ of the states. Recently, Pelecanos, Tan, Tang, and Wright~\cite{pelecanos2025beating} proposed a nonlinear extension of classical shadow estimation to estimate $\T(\rho^k)$ using a natural unbiased estimator motivated by U-statistics. Suppose we are given $N$ copies of a quantum state $\rho$ and a fixed positive integer $k \leq N$. For each copy, a uniform POVM is performed, and the outcome $\ket{u_i}$ on the $i$-th copy is used to define the associated observable $\hat{\sigma}_i := (d + 1)\ket{u_i}\bra{u_i} - I$. An unbiased estimator for $\T(\rho^k)$ is given by
\begin{equation}
    Z_k := \frac{1}{N^{\underline{k}}} \sum_{\substack{i_1, \dots, i_k \in [n] \\ \text{distinct}}} \T\left( \hat{\sigma}_{i_1} \cdots \hat{\sigma}_{i_k} \right),
\end{equation}
where $N^{\underline{k}} := N(N - 1)\cdots(N - k + 1)$ denotes the falling factorial. For any quantum state $\rho$ of dimension $d$, and any fixed $k \geq 2$, it was shown that with probability at least $0.99$, the estimator $Z_k$ approximates $\T(\rho^k)$ up to a multiplicative error $\epsilon$, using $N = \mathcal{O} ( \max\{ {d^{2 - 2/k}}/{\epsilon^2}, {d^{3 - 2/k}}/{\epsilon^{2/k}} \} )$ copies of $\rho$. Naturally, this results in exponential scaling with respect to the number of qubits. Note that their estimator provides a multiplicative-error approximation for $\T(\rho^k)$, which immediately yields an additive-error estimator for the quantum Rényi entropy. In contrast, our work focuses on estimating $\T(\rho^k)$ up to an additive error, which in turn implies an additive-error estimator for the quantum Tsallis entropy.

Several studies have explored the relationship between the trace of quantum state powers, quantum entanglement, and separability testing. Among them, Bradshaw \textit{et al.}~\cite{bradshaw2023cycle} investigates quantum separability tests from the perspective of combinatorial group theory, uncovering a fundamental link between the acceptance probabilities of these tests and the cycle index polynomials of finite groups. The cycle index polynomial of a permutation group $\mathcal{G}$ is defined as
\begin{equation}
    Z(\mathcal{G})(x_1,\dots,x_n):=\frac{1}{\abs{\mathcal{G}}}\sum_{g\in \mathcal{G}} \prod_{i=1}^n x_i^{c_i(g)},
\end{equation}
where $c_j(g)$ represents the number of cycles of length $j$ in the disjoint cycle decomposition of $g$. Notably, in the generalization of the bipartite pure-state separability algorithm, the acceptance probability associated with a group $\mathcal{G}$, denoted as $p_\mathcal{G}$, takes the form
\begin{equation}
    p_\mathcal{G} = Z(\mathcal{G})(1, \dots, \T(\rho^k)).
\end{equation}
This implies that $p_\mathcal{G}$ is determined by evaluating the cycle index polynomial of $\mathcal{G}$ at $x_j = \T(\rho^j)$ for $j \in \{1, \dots, k\}$. The study first derives an exact analytical expression for the probability of a mixedness test accepting as the number of state copies increases, showing that this probability is governed by the cycle index polynomial of the symmetric group. Building on this insight, the authors extend the framework to develop a family of separability tests corresponding to arbitrary finite groups, proving that the acceptance probability aligns with the cycle index polynomial of the respective group. Furthermore, they propose explicit quantum circuit implementations for these tests, leveraging \textsf{CSWAP} gates in a resource-efficient manner—scaling as $\mathcal{O}(k^2)$ for the symmetric group and $\mathcal{O}(k \ln (k))$ for the cyclic group, where $k$ denotes the number of state copies used in the test. The study of partial transpose moments and entanglement detection was discussed in the work by Neven \textit{et al.}~\cite{neven2021symmetry}, and~\cref{chap:entangle} provides a more detailed discussion on how our work can be applied to that research. Additionally, Wagner \textit{et al.}~\cite{wagner2024quantum} proposed simple quantum circuits for measuring weak values, Kirkwood–Dirac (KD) quasiprobability distributions, and the spectra of quantum states without post-selection, particularly by interpreting the trace of quantum state powers from the perspective of measuring unitary-invariant and relational properties of quantum states using Bargmann invariants.

Finally, the work by Liu and Wang~\cite{liu2025estimating}, published several months after the first version of our paper appeared on arXiv, provides a detailed computational complexity analysis of estimating the trace of powers of an $n$-qubit mixed quantum state $\rho$, given a state-preparation circuit of size $\text{poly}(n)$. Leveraging efficiently computable uniform approximations of positive power functions within the framework of quantum singular value transformation, the authors achieved an exponential improvement over previously known methods. Their study focused particularly on estimating the quantum Tsallis entropy, \begin{equation} S_k(\rho) = \frac{1 - \text{Tr}(\rho^k)}{k-1}, \end{equation} and precisely identified the thresholds at which the computational complexity of the problem undergoes qualitative changes. Specifically, they showed that for $k = 1$, the problem is $\textsf{NIQSZK-complete}$; for $1 < k \leq 1 + {(n-1)}^{-1}$, it is $\textsf{NIQSZK-hard}$; for $1 + \Omega(1) \leq k \leq 2$, it becomes $\textsf{BQP-complete}$; and for $k > 2$, it remains within $\textsf{BQP}$. They also established rigorous bounds on the query and sample complexity across different regimes of $k$, with particular attention to rank-dependent behavior. In their publication, the authors characterized the method proposed in our initial preprint as a rank-dependent estimator for the quantum Tsallis entropy in the regime where $k$ exceeds the rank of the quantum state. Our current results further strengthen this interpretation, as our method now provides an $\epsilon$-additive estimator whenever $k$ exceeds the effective rank $\widetilde{r} = \min\left\{r, \left\lceil\ln\left({2k}/{\epsilon}\right)\right\rceil\right\}$.

\section{Iterative algorithm for estimating the trace of quantum state powers}
In this section, we explain the Newton-Girard method and discuss how it is utilized in the design of our algorithm. In proving the main theorems, the key idea underlying our improvement is the use of the Newton–Girard identities, which establish an explicit relationship between power sums (e.g., $x_1^k + x_2^k$) and elementary symmetric polynomials (e.g., $x_1 + x_2$, $x_1x_2$). By leveraging these identities alongside a careful and systematic analysis, we can efficiently express higher-order moments in terms of lower-order symmetric functions, thereby reducing the overall estimation complexity. This connection provides a clear algebraic intuition behind our approach and highlights why fewer moment estimations suffice.

\subsection{Intuition: Newton-Girard method}\label{chap:NG}
The main idea of entanglement spectroscopy demonstrates that the trace of quantum state powers can be used to estimate the largest eigenvalues~\cite{johri2017entanglement, yirka2021qubit, song2012bipartite}. The $k$ largest eigenvalues can be estimated using $\{\T(\rho^i)\}_{i=1}^k$. The Newton-Girard method~\cite{bagdasaryan2015analogues, cereceda2022sums} provides the mathematical foundation of entanglement spectroscopy and serves as an important component in our method. Therefore, we describe the details of the Newton-Girard method and explain the inspiration that leads to the notion that \textit{`rank is sufficient'} for estimating the trace of quantum state powers.

Let $r=\text{rank}(\rho)$, and the eigenvalues of $\rho$ are $\{p_i\}_{i=1}^r$, sorted in descending order. {We utilize the Newton-Girard method to leverage the following well-known result from linear algebra and provide an intuition for it: knowing the trace of quantum state powers $\{\T(\rho^i)\}_{i=1}^r$ is equivalent to knowing $\{p_i\}_{i=1}^r$.} Consider the equation having these eigenvalues as root in the form of
\begin{equation}\label{eq5}
    \prod_{m=1}^{r}(x-p_m) = 0.
\end{equation}
The values of $\mathrm{Tr}(\rho^i)$ are now the $i$-th power sum of the roots. Denote the power sum as
\begin{equation}
    P_i := \sum_{m=1}^{r}p_m^i = \mathrm{Tr}(\rho^i).
\end{equation}
Here, Simply expanding the terms of the~\cref{eq5} above as follows:
\begin{equation}\label{eq:7}
    \prod_{m=1}^{r}(x-p_m) = \sum_{k=0}^{r} (-1)^k a_k x^{r-k},
\end{equation}
where $a_k$ is the elementary symmetric polynomial, defined as the sum of all distinct products of $k$ distinct variables, such as:
\begin{align*}
    a_0 &= 1, \\
    a_1 &= p_1 + p_2 + \ldots + p_r = \sum_{1 \leq \alpha \leq r} p_\alpha, \\
    a_2 &= p_1 p_2 + p_1 p_3 + \ldots + p_{r-1} p_r = \sum_{1 \leq \alpha < \beta \leq r} p_\alpha p_\beta, \\
    a_3 &= \sum_{1 \leq \alpha < \beta < \gamma \leq r} p_\alpha p_\beta p_\gamma, \\
    \vdots \\
    a_r &= p_1 p_2 \ldots p_r.
\end{align*}

The Newton-Girard method states the relationship of the elementary symmetric polynomials and the power sums recursively as follows. For all $r\ge k \ge 1$,
\begin{equation}\label{elepoly1}
    a_k = \frac{1}{k}\sum_{i=1}^{k} (-1)^{i-1} a_{k-i} P_i.
\end{equation}
Given $ P_i $ for $ 1 \leq i \leq r $, we can uniquely determine the values of $ a_k $ on the right-hand side of~\cref{eq:7}. Moreover, the set of eigenvalues $ \{p_i\}_{i=1}^r $ is also uniquely determined as the roots of~\cref{eq5}.

Unfortunately in real-world situations, we cannot exactly calculate the trace of quantum state powers; instead, we can obtain the estimation with errors using previous strategies. Then, it is natural to ask the following question:
\begin{center}
    \textit{``If the error of estimated power sums is small, are the roots obtained by the Newton-Girard method close to the eigenvalues of $\rho$?''}
\end{center}
No, the statement is not always true. A counterexample is Wilkinson's polynomial~\cite{mosier1986root}, which shows that the location of the roots can be very sensitive to perturbations in the coefficients of a polynomial. Generally, to obtain the eigenvalues, the estimation error of the trace of quantum state powers should be exponential, causing the copy and time complexity to be exponential~\cite{johri2017entanglement}. Therefore, estimating the eigenvalues with the estimated values of $\{\T(\rho^i)\}_{i=1}^r$ is unfeasible. 

However, we get an intuition from the Newton-Girard method that estimating $\{\T(\rho^i)\}_{i=1}^r$ contains valuable information about the quantum states. In~\cref{chap:rank}, we prove that estimating the trace of quantum state powers $\{\T(\rho^i)\}_{i=1}^r$ is sufficient for estimating the trace of larger powers $\T(\rho^i)$ for $i>r$. The error of each eigenvalue obtained by the Newton-Girard method is large, but as the power of the eigenvalues is summed up, the error diminishes to a smaller extent.

\subsection{Explicit algorithm construction}\label{chap:algo}
\begin{figure*}[htbp!]
  \centering
  \includegraphics[width=15cm]{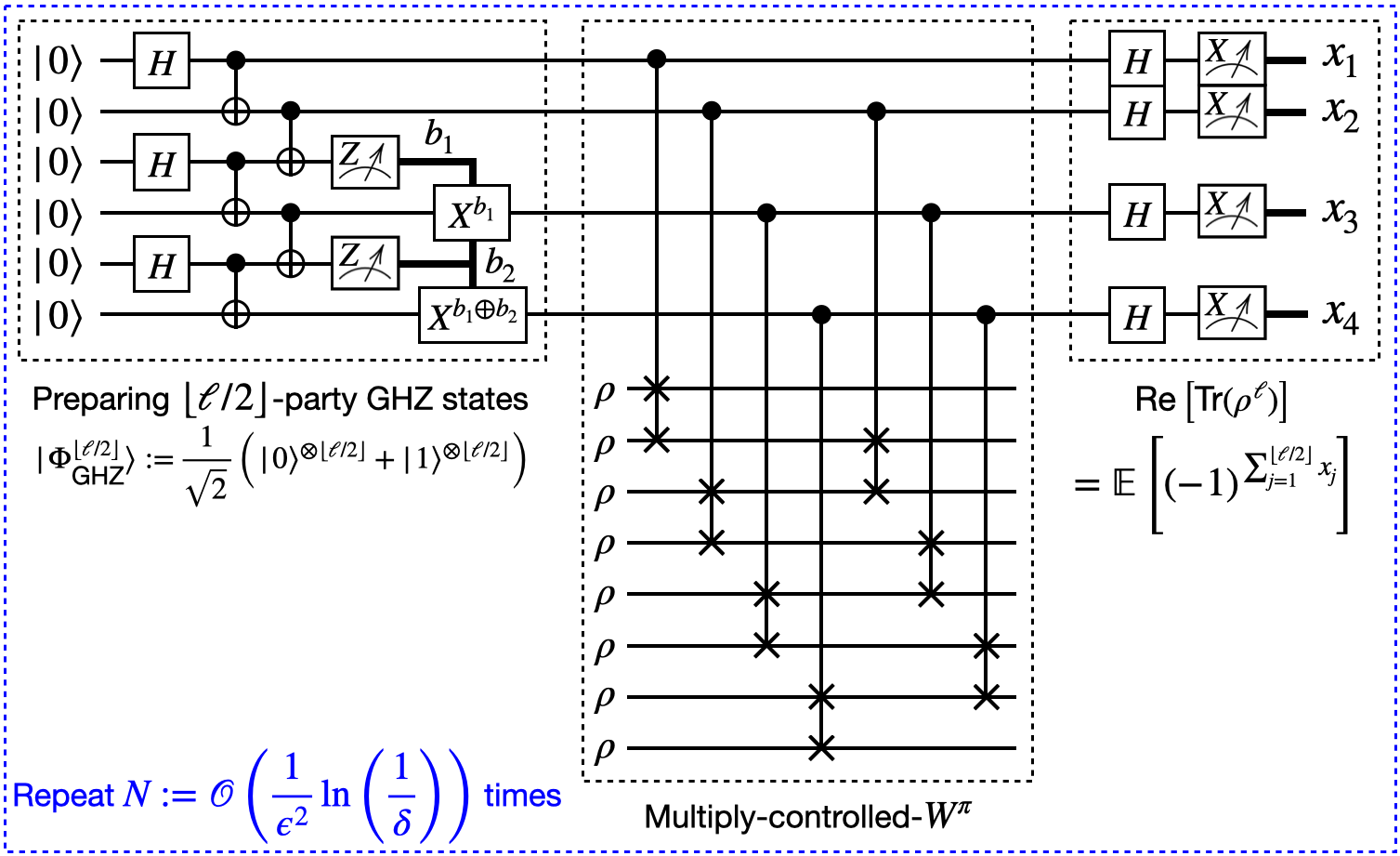}
  \caption{{\textbf{Quantum circuit for Step 1 of the {[Algorithm 1]}: A detailed example for $\ell = 8$.} This quantum circuit is used to estimate $\text{Tr}(\rho^\ell)$. The first part of the circuit corresponds to the GHZ state preparation described in Step 1(a). As mentioned in the main text, this step can be implemented differently if necessary. Following this, a multiply-controlled cyclic shift operation is applied, with slight structural variations depending on whether $\ell \mod 2$ is 0 or 1. The type of gate applied before measurement and the measurement basis used depend on whether $\text{Re}[\text{Tr}(\rho^\ell)]$ or $\text{Im}[\text{Tr}(\rho^\ell)]$ is being estimated. By calculating the expectation of the measured outcomes, the desired physical quantity can be estimated. To ensure a good estimate with an additive error of at most $\epsilon$ with high probability $1 - \delta$, as guaranteed by~\cref{eq:gua}, $\mathcal{O}(\ln({1}/{\delta})/{\epsilon^2})$ repetitions of the steps within the blue box in the figure are required.}}
  \label{fig:multivariate}
\end{figure*}

\begin{figure*}[htbp!]
  \centering
  \includegraphics[width=11cm]{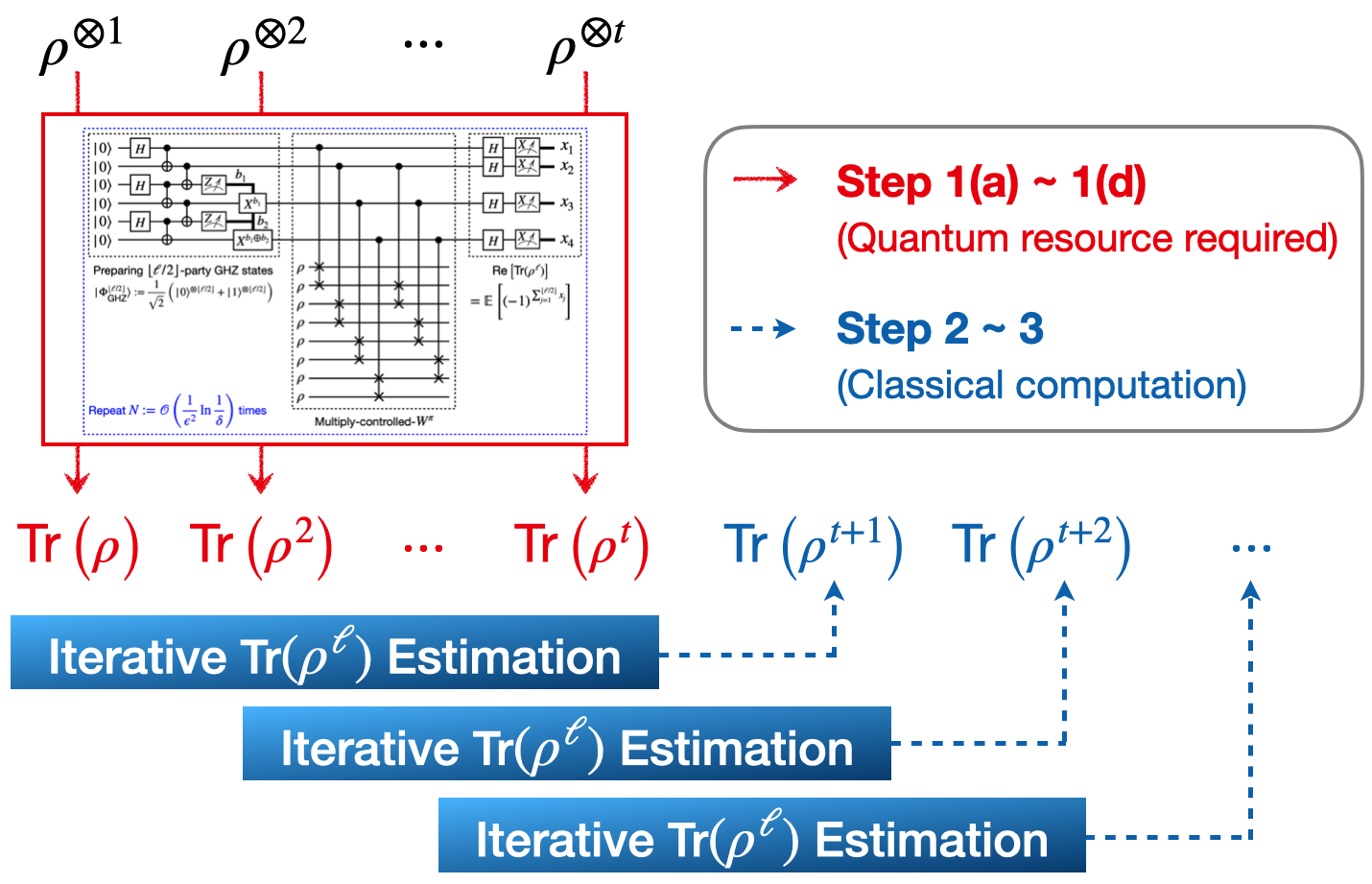}
  \caption{{\textbf{Diagram of the complete process of the {[Algorithm 1]}.} The red box represents the process shown in~\cref{fig:multivariate}, corresponding to Step 1 of the algorithm described in~\cref{chap:algo}. Quantum resources are required only for this step, during which the values from $\text{Tr}(\rho)$ to $\text{Tr}(\rho^t)$ are obtained. The subsequent blue dashed lines indicate computations performed using a simple recurrence relation without requiring quantum resources, following the processes outlined in Steps 2 and 3.}}
  \label{fig:algo}
\end{figure*}

Based on the insights gained in~\cref{chap:NG}, the specific algorithm for calculating the trace of quantum state powers is as follows. (In this section, the index is denoted by $\ell$ to avoid confusion with $i = \sqrt{-1}$, which represents the imaginary unit.)\\~\\
\underline{{\textbf{[Algorithm 1]} Estimation of $\T(\rho^k)$}}
\begin{enumerate}\label{algo:1}
\item Based on the circuit presented in~\cref{fig:multivariate}, the values of $\text{Tr}(\rho^\ell)$ for $\ell = 1, 2, \dots, t$ are obtained. This process is based on research on multivariate trace estimation using constant quantum depth~\cite{quek2024multivariate}, and the detailed procedure is as follows: (An example for $\ell = 8$ is illustrated in~\cref{fig:multivariate} for reference.)
    \begin{enumerate}
        \item Generate an $\lfloor \ell/2 \rfloor$-party GHZ state
        \begin{equation}
            |\Phi_\text{GHZ}^{\lfloor \ell/2 \rfloor}\rangle:=\frac{1}{\sqrt{2}}\left(|0\rangle^{\otimes{\lfloor \ell/2 \rfloor}}+|1\rangle^{\otimes{\lfloor \ell/2 \rfloor}}\right).
        \end{equation}
        This process utilizes mid-circuit measurement and requires a constant quantum-depth circuit along with classical feedback, while a logarithmic-depth classical circuit is needed for parity computation. Besides the method illustrated in~\cref{fig:multivariate}, it is also possible to employ other methods for generating a GHZ state.
        \item Next, a multiply-controlled cyclic shift operation is performed. Depending on whether $\ell$ is odd or even, slight structural modifications to the circuit may be necessary. The specifics are discussed in detail in of~{\cite[Section 3.2]{quek2024multivariate}}, and this process can be achieved with constant quantum depth.
        \item To estimate $\text{Re}[\text{Tr}(\rho^\ell)]$, apply an $H$ gate to all $\lfloor \ell/2 \rfloor$ qubits and measure in the $X$-basis.\\
        \textbf{Note:} While the multivariate trace estimation problem in the original study requires the estimation of both the real and imaginary parts, our problem focuses solely on estimating the trace of quantum powers, making the estimation of the real part sufficient. However, for the sake of completeness, we also describe the process for estimating the imaginary part: replace the $H$ gate with an $HS^\dagger$ gate and measure in the $Y$-basis to estimate $\text{Im}[\text{Tr}(\rho^\ell)]$.
        \item Repeat the process from Step 1(a) to 1(c) 
            \begin{equation}
                N:=\mathcal{O}\left(\frac{\ln(1/\delta)}{\epsilon^2}\right)
            \end{equation}
         times, and let the measurement outcomes (0 or 1) obtained in Step 1(c) for the $m$-th iteration be denoted as $x_1^m, \ldots, x_{\lfloor \ell/2 \rfloor}^m, y_1^m, \ldots, y_{\lfloor \ell/2 \rfloor}^m$. Then, the quantity we aim to estimate, $\text{Tr}(\rho^\ell)$, is expressed as:
        \begin{align}
            & Q_{\ell(\le t)} := \hat{\mathcal{R}} + i \hat{\mathcal{J}} \approx \text{Tr}(\rho^\ell), \\
            & \text{where}~\hat{\mathcal{R}} = \frac{\sum_{m=1}^N \sum_{j=1}^{\lfloor \ell/2 \rfloor} (-1)^{x_j^m}}{N},\quad\text{and}~\hat{\mathcal{J}} = \frac{\sum_{m=1}^N \sum_{j=1}^{\lfloor \ell/2 \rfloor}, (-1)^{y_j^m}}{N}.
        \end{align}
    Then this estimate satisfies the inequality below for $\ell = 1, 2, \ldots, t$.
    \begin{equation}\label{eq:gua}
        \text{Pr}\left(\abs{Q_{\ell} - \text{Tr}(\rho^\ell)}\le \epsilon\right) \ge 1-\delta.
    \end{equation} 
    \end{enumerate}
    \textbf{Note:} As mentioned in Step 1(c), $\text{Tr}(\rho^\ell)\in\mathbb{R}$, so it does not matter if we set $Q_{\ell(\le t)} = \hat{\mathcal{R}}$. This is because, through this algorithm, we have 
    \begin{equation}
        \sqrt{\left(\hat{\mathcal{R}} - \text{Tr}(\rho^\ell)\right)^2 + {\hat{\mathcal{J}}}^2} \leq \epsilon,   
    \end{equation}
     which also ensures that $\abs{\hat{\mathcal{R}} - \text{Tr}(\rho^\ell)} \leq \epsilon$.
\item Calculate the elementary symmetric polynomial $b_k~(1\le k\le t)$ defined as:
    \begin{equation}\label{elepoly2}
        b_k = \frac{1}{k}\sum_{\ell=1}^{k} (-1)^{\ell-1} b_{k-\ell} Q_\ell,~b_0=1.
    \end{equation}
\item Using $Q_1, \ldots, Q_t$ obtained from Step 1 and $ b_1, \ldots, b_t $ obtained from Step 2, the value of $\text{Tr}(\rho^\ell)~(\ell > t)$ can be estimated through the following recurrence relation:
    \begin{equation}\label{DefQ}
        Q_{\ell (> t)} := \sum_{k=1}^{t} (-1)^{k-1} b_k Q_{\ell-k} \approx \text{Tr}(\rho^\ell).
    \end{equation}
\end{enumerate}

Through Step 3, we can obtain values for $Q_{t + 1},  Q_{t + 2}, \ldots $, and in~\cref{chap:rank} and~\ref{chap:effective}, we analyze in detail the conditions on $t$ required to ensure that the estimated values obtained in this process are within an additive error of at most $\epsilon$. See~\cref{fig:algo} for the overall process of the algorithm we propose. Note that quantum devices are only used to estimate $\{\T(\rho^i)\}_{i=1}^t$. At most $\mathcal{O}(t)$ qubits and $\mathcal{O}(t)$ multi-qubit gates are required (used only in Step 1).

\section{Analysis of the proposed algorithm}\label{chap:3}
We analyze our proposed algorithm in two phases.
\begin{enumerate}[label=(\arabic*)]
    \item In~\cref{chap:rank}, we show that $t \geq r$ is sufficient, identifying the rank dependence.
    \item Then, in~\cref{chap:effective}, we prove that $t \ge \left\lceil\ln\left({2k}/{\epsilon}\right)\right\rceil$ is sufficient for estimating trace of quantum state powers within an additive error of $\epsilon$.
\end{enumerate}
This introduces the new concept of the effective rank, leading to stronger results and enabling the algorithm to be applicable even when the exact rank is unknown. In~\cref{chap:obs}, we discuss the case that includes arbitrary observables, which is a more generalized version of the problem of estimating the trace of quantum state powers.

For clarity, we summarize the notations used in this section. Let ${P_i}$ represent the \textit{exact values} of the trace of quantum state powers:
\begin{equation}\label{def-P}
    P_i := \T(\rho^i) = \sum_{j=1}^r p_j^i.
\end{equation}
Similarly, let ${Q_i}$ denote the \textit{estimated values} of the trace of quantum state powers. For $Q_{i (\leq t)}$, the values are estimated by the quantum device, while for $Q_{i (> t)}$, they are defined by~\cref{DefQ}. The estimation (additive) error is denoted as
\begin{equation}
    \epsilon_i := Q_i - P_i.    
\end{equation}

Next, let $a_k$ and $b_k$ represent the elementary symmetric polynomials corresponding to ${P_i}$ and ${Q_i}$, respectively (see~\cref{elepoly1} and~\cref{elepoly2}). In~\cref{lem:rank}, we analyze the bound on the difference between these two elementary symmetric polynomials.

In cases where quantum resources are so limited that even utilizing resources commensurate with the rank is infeasible, we may not be able to estimate all elements of $\{\text{Tr}(\rho^i)\}_{i=1}^{r}$ but only up to $\{\text{Tr}(\rho^i)\}_{i=1}^{t \, (<r)}$. To account for this limitation, we introduce a new quantity, denoted as $\widetilde{P}$, which is defined as follows for $i \leq t$:  
\begin{equation}\label{p'}
    \widetilde{P}_{i(\leq t)} := \text{Tr}(\rho^{i}) = \sum_{j=1}^r p_{j}^i.
\end{equation}  
For $ i > t $, $ \widetilde{P}_{i} $ is recursively defined based on the Newton-Girard recurrence relations, where the elementary symmetric polynomials $a_k$ are identical to $P_i$: (Since $\widetilde{P} = P$ when $ t = r $, the recurrence relation can also be applied to $ P $ in this case.)  
\begin{equation}\label{p'recurr}  
    \widetilde{P}_{i(> t)} := \sum_{k=1}^t (-1)^{k-1} a_k \widetilde{P}_{i-k}.  
\end{equation}  
While the introduction of $\widetilde{P}$ may appear indistinguishable from the original definition of $P$ in~\cref{def-P}, a fundamental distinction arises when $t < r$, as $\widetilde{P}_{i(>t)} \neq P_{i(>t)}$. The concept of $\widetilde{P}$ is particularly useful for quantifying the impact of information loss on error when only partial spectral information is available, and in~\cref{lem:eff}, we provide a rigorous quantitative analysis of the discrepancy between $\widetilde{P}$ and $P$.

For the problem of the trace of quantum state powers with arbitrary observables, given a quantum state
\begin{equation}
    \rho = \sum_{i=1}^r p_i \ket{\psi_i} \bra{\psi_i}    
\end{equation}
and an arbitrary observable $M$, we can define $P_{i,M}$ similarly to $P_i$ as follows:  
\begin{equation}
    P_{k,M} := \mathrm{Tr}(M\rho^k) = \sum_{i=1}^r \bra{\psi_i} M \ket{\psi_i} p_i^k.    
\end{equation}
Likewise, $Q_{i,M}$ represents its estimated value.

\subsection{Rank is all you need}\label{chap:rank}
In this section, we prove that $t \geq r$ is sufficient to proceed with the algorithm while maintaining low error, and we derive the required number of quantum circuit runs. Although the method is simple, we argue that it offers advantages in terms of the number of required qubits and multi-qubit gates. To the best of our knowledge, our work is the first to prove that the traces of rank-at-most-$r$ powers, $\{\T(\rho^i)\}_{i=1}^r$, are sufficient for estimating $\T(\rho^k)$ when $k$ is large. Furthermore, it provides an efficient algorithm, particularly for low-rank quantum states. Our goal is to first establish a quantitative bound on the difference between the elementary symmetric polynomials derived from the true values $P_i$ and those obtained from the estimated values $Q_i$.
\begin{lemma}\label{lem:rank}
Let $d_k := b_k - a_k$, then the following holds:
\begin{equation}
    \abs{d_k} \leq \sum_{j=1}^k \frac{\abs{\epsilon_j}}{j}.
\end{equation}
\end{lemma}
\begin{proof}
     See the details in~\cref{lem:rank:proof}.
\end{proof}

Now, we prove our first theorem, which demonstrates that $t = r$ is sufficient to execute our algorithm with low error.
\begin{theorem}\label{thm:rank}
Suppose that,
\begin{equation}\label{eq:esti}
    \varepsilon_i:=\abs{\epsilon_i} = \abs{Q_i - P_i} < \frac{\epsilon}{kt\ln t}
\end{equation}
holds for $i=1, 2, \ldots, t$. {Setting $t = r$ and proceeding with {\textbf{[Algorithm 1]}} based on the recurrence relation~\cref{DefQ}, the following relation always holds:}
\begin{equation}
    \abs{\epsilon_i} = \abs{Q_i - P_i} < \epsilon
\end{equation}
for $i=t+1, \ldots, k$.
\end{theorem}
\begin{proof}
     See the details in~\cref{thm:rank:proof}.
\end{proof}
Based on~\cref{thm:rank}, the quantum resources required to solve the problem of estimating the trace of quantum state powers are derived in~\cref{cor:rank}.
\begin{corollary}\label{cor:rank}
To estimate $\T(\rho^i)$ for all $i \leq k$ within an additive error of $\epsilon$ and with a success probability of at least $1 - \delta$, where $~\delta \in (0,1)$, it suffices to estimate each $\T(\rho^j)$ for $j \leq r$ within an additive error of $\varepsilon_j$, as defined in~\cref{thm:rank}. This can be achieved by using
\begin{equation}\label{eq:runs}
    \mathcal{O}\left(\frac{k^2r^2\ln^2 r\ln(1/\delta)}{\epsilon^2} \right)
\end{equation}
runs on a constant-depth quantum circuit consisting of $\mathcal{O}(j)$ qubits and $\mathcal{O}(j)$ \textsf{CSWAP} operations.
\end{corollary}
\begin{proof}
     See the details in~\cref{cor:rank:proof}.
\end{proof}
Using a quantum device, $j$ copies of $\rho$ are required for each run of the quantum circuit to estimate $\T(\rho^j)$. The number of runs is the same for every $\T(\rho^j)$, as specified in~\cref{eq:runs}. Since $j \leq r$, the total number of copies needed to estimate $\{\T(\rho^i)\}_{i=1}^k$ within an additive error of $\epsilon$ is
\begin{equation}\label{eq:runs-cal}
    \mathcal{O}\left(\sum_{j=1}^r \frac{j k^2r^2}{\epsilon^2}\ln^2 r \right) = \mathcal{O}\left(\frac{k^2r^4}{\epsilon^2}\ln^2 r \right).
\end{equation}

To highlight the significance of our work and aid understanding, we present the following proposition, which provides a simplified version of~\cref{thm:rank} and~\cref{cor:rank}.
\begin{proposition}[Informal, see~\cref{thm:rank} and~\cref{cor:rank}]\label{prop:rank}
In {\textbf{[Algorithm 1]}}, setting $t \geq r$ is sufficient to efficiently estimate the trace of quantum state powers, even for large powers. (In other words, the problem of estimating the trace of quantum state powers requires quantum resources proportional to the rank of the given quantum state, rather than its power $k$.)
\end{proposition}

\subsection{Effective rank is all you need}\label{chap:effective}
%=============================
\begin{table*}[ht]
\centering
\resizebox{\textwidth}{!}{
\begin{tabular}{>{\centering\arraybackslash}m{4cm}|c|c|c|c|c}
\hline
\textbf{Method} & \textbf{\# Depth}  & \textbf{\# Qubits}  & \textbf{\# \textsf{CSWAP}} & \textbf{\# Copies} & \textbf{Original $|\psi\rangle$} \\ \hline\hline
Generalized swap test  \cite{ekert2002direct} & $\mathcal{O}(k)$ & $\mathcal{O}(k)$ & $\mathcal{O}(k)$ & $\mathcal{O}\left({k^2}/{\epsilon^2}\right)$ & {\textbf{NOT}} required \\\hline
Hadamard  test \qquad\quad \cite{johri2017entanglement} & $\mathcal{O}(k)$ & $\mathcal{O}(k)$ & $\mathcal{O}(k)$ & $\mathcal{O}\left({k^2}/{\epsilon^2}\right)$ & {Required} \\\hline
Two-copy test  \qquad\quad\cite{subacsi2019entanglement} & $\mathcal{O}(1)$ & $\mathcal{O}(k)$  & $\mathcal{O}(k)$ & $\mathcal{O}\left({k^2}/{\epsilon^2}\right)$ & {Required} \\\hline
Two-copy test \& Qubit-reset~\cite{yirka2021qubit}   & $\mathcal{O}(k)$ & $\mathcal{O}(1)$ & $\mathcal{O}(k)$ & $\mathcal{O}\left({k^2}/{\epsilon^2}\right)$ & Required \\\hline
Multivariate trace estimation~\cite{quek2024multivariate}  & $\mathcal{O}(1)$ & $\mathcal{O}(k)$ & $\mathcal{O}(k)$ & $\mathcal{O}\left({k^2}/{\epsilon^2}\right)$ & {\textbf{NOT}} required \\\hline
Ours \qquad\quad\qquad\quad (this work) & $\mathcal{O}(1)$ & $\mathcal{O}(\widetilde{r})$ & $\mathcal{O}(\widetilde{r})$ & $\widetilde{\mathcal{O}}\left({k^2}/{\epsilon^2}\right)$ & {\textbf{NOT}} required \\ \hline
\end{tabular}}
\caption{\textbf{Summary of resources required by different algorithms to estimate the values of $\{\T(\rho^i)\}_{i=1}^k$ within an error margin of $\epsilon$.} The comparison includes a total of six algorithms, including ours. The algorithms are categorized based on quantum circuit depth, the number of required qubits, the number of required \textsf{CSWAP} operations, the number of required quantum states $\rho$, and whether the original state $|\psi\rangle$ is needed for the algorithm to operate. Here, the notation $\widetilde{\mathcal{O}}(\cdot)$ hides polylogarithmic factors in $k$, and $\widetilde{r} = \min\left\{r, \left\lceil\ln\left({2k}/{\epsilon}\right)\right\rceil\right\}$ is the effective rank defined in~\cref{eq:eff-rank}.}
\label{tab:compare}
\end{table*}
%=============================

In~\cref{chap:rank}, we identified for the first time how the complexity of estimating $\T(\rho^k)$ can improve when the quantum state $\rho$ has low rank. However, this advantage critically relies on the exact knowledge of the rank. In practice, such knowledge is rarely available, and even when the rank is known, the improvement may be marginal unless the state is \textit{very} low-rank compared to its dimension. This motivates a broader question: under what conditions can we still benefit from our algorithm when the rank is unknown, approximately known, or when quantum resources are limited?

To address this, we move away from analyses that require precise knowledge of the rank of $\rho$, and instead seek a more nuanced understanding of the complexity in terms of the parameters $k$ and $\epsilon$, which govern the exponent in $\T(\rho^k)$ and the target additive precision. Although the rank can have a noticeable impact when it is very small, its effect diminishes rapidly as the state becomes more full-rank. In such cases, the value of $\T(\rho^k)$ often becomes so small that it cannot be meaningfully distinguished from zero within realistic precision bounds. This motivates us to treat $k$ and $\epsilon$ as the central parameters driving the complexity, rather than relying on exact rank information.

This observation naturally leads to the notion of an effective rank, which captures how many eigenvalues of $\rho$ make a meaningful contribution to $\T(\rho^k)$ within the desired precision. Rather than focusing solely on the full rank of $\rho$, the effective rank offers a more nuanced and practical understanding of when the algorithm remains useful in realistic settings. It emphasizes that the true complexity is more fundamentally governed by the interplay between $k$ and the target accuracy $\epsilon$, rather than the sheer dimension of the system.

First, let us examine the quantitative difference between $\widetilde{P}$ and $P$.
\begin{lemma}\label{lem:eff}
Suppose that $\widetilde{P}_{i}$ is define by~\cref{p'},~\cref{p'recurr}. Then the following holds:
\begin{equation}
    \abs{\widetilde{P}_k-P_k} \le \frac{k}{t!}\left(1-\frac{t}{r}\right).
\end{equation}
\end{lemma}
\begin{proof}
    See the details in~\cref{lem:eff:proof}.
\end{proof}
Now, we prove our second theorem, which demonstrates that $t = \left\lceil\ln\left({2k}/{\epsilon}\right)\right\rceil$ is sufficient to execute our algorithm with low error.
\begin{theorem}\label{thm:eff}
Suppose that,
\begin{equation}\label{eq:esti-thm3}
    \varepsilon_i:=\abs{\epsilon_i} = \abs{Q_i - P_i} < \frac{\epsilon}{2kt\ln t}
\end{equation}
holds for $i=1, 2, \ldots, t$. {Setting $t = \left\lceil\ln\left({2k}/{\epsilon}\right)\right\rceil$ and proceeding with {\textbf{[Algorithm 1]}} based on the recurrence relation~\cref{DefQ}, the following relation always holds:}
\begin{equation}
    \abs{\epsilon_i} = \abs{Q_i - P_i} < \epsilon
\end{equation}
for $i=t+1, \ldots, k$.
\end{theorem}
\begin{proof}
    See the details in~\cref{thm:eff:proof}.
\end{proof}

Based on~\cref{thm:eff}, the quantum resources required to solve the problem of estimating the trace of quantum state powers are derived in~\cref{cor:eff}.
\begin{corollary}\label{cor:eff}
To estimate $\T(\rho^i)$ for all $i \leq k$ within an additive error of $\epsilon$ and with a success probability of at least $1 - \delta$, where $~\delta \in (0,1)$, it suffices to estimate each $\T(\rho^j)$ for $j \leq \left\lceil\ln\left({2k}/{\epsilon}\right)\right\rceil$ within an additive error of $\varepsilon_j$, as defined in~\cref{thm:eff}. This can be achieved by using
\begin{equation}\label{eq:runs-eff}
    \widetilde{\mathcal{O}}\left(\frac{k^2\ln(1/\delta)}{\epsilon^2}\right)
\end{equation}
runs on a constant-depth quantum circuit consisting of $\mathcal{O}(j)$ qubits and $\mathcal{O}(j)$ \textsf{CSWAP} operations. Here, the notation $\widetilde{\mathcal{O}}(\cdot)$ hides polylogarithmic factors in $k$.
\end{corollary}
\begin{proof}
     It follows the same logic as the proof of~\cref{cor:rank}. Please refer to~\cref{cor:rank:proof}.
\end{proof}

Using a quantum device, $j$ copies of $\rho$ are required for each run of the quantum circuit to estimate $\T(\rho^j)$. The number of runs is the same for every $\T(\rho^j)$, as specified in~\cref{eq:runs-eff}. Since $j \leq \left\lceil\ln\left({2k}/{\epsilon}\right)\right\rceil$, the total number of copies needed to estimate $\{\T(\rho^i)\}_{i=1}^k$ within an additive error of $\epsilon$ is
\begin{equation}
    \widetilde{\mathcal{O}}\left(\frac{k^2}{\epsilon^2}\right),
\end{equation}
where $\widetilde{\mathcal{O}}(\cdot)$ hides polylogarithmic factors in $k$.

Again, to highlight the significance of our work and aid understanding, we present the following proposition, which provides a simplified version of~\cref{thm:eff} and~\cref{cor:eff}.
\begin{proposition}[Informal, see~\cref{thm:eff} and~\cref{cor:eff}]\label{prop:eff}
In {\textbf{[Algorithm 1]}}, setting $ t \geq \left\lceil\ln\left({2k}/{\epsilon}\right)\right\rceil $ is sufficient to efficiently estimate the trace of quantum state powers $ \{\T(\rho^i)\}_{i=1}^k $ with an additive error of at most $ \epsilon $. (In other words, the problem of estimating the trace of quantum state powers requires quantum resources proportional to the logarithm of the number of powers, rather than the power $k$.)
\end{proposition}

To conclude~\cref{chap:rank} and~\ref{chap:effective}, we summarize our findings in the following theorem:  
\begin{proposition}[Informal description of the main results]\label{thm:main-1}  
For the problem of estimating the trace of quantum state powers, given a large integer $ k $, it is possible to approximate $ \{\T(\rho^i)\}_{i=1}^k $ within an additive error of $ \epsilon $ using quantum resources only up to $ \{\T(\rho^i)\}_{i=1}^t $, where $ t $ is given by 
\begin{equation}\label{eq:eff-rank}
    t \ge \widetilde{r} = \min\left\{r, \left\lceil\ln\left(\frac{2k}{\epsilon}\right)\right\rceil\right\}.
\end{equation}
We define $ \widetilde{r} $ as the effective rank.  
\end{proposition}

In this way, we present a strengthened result from~\cref{chap:rank}, incorporating the concept of effective rank to achieve a more refined analysis.

As mentioned, our work provides an advantage in terms of the number of needed qubits and multi-qubit gates. Since we only need to estimate $\{\T(\rho^i)\}_{i=1}^{\tilde{r}}$, $n\widetilde{r}$ qubits and $\mathcal{O}(\widetilde{r})$ \textsf{CSWAP} operations are sufficient for the estimation. We emphasize that reducing the number of qubits and \textsf{CSWAP} operations used in the quantum circuit is an important improvement because it is less sensitive to noise, and having fewer qubits is advantageous for implementation on near-term quantum devices~\cite{corcoles2019challenges, georgopoulos2021modeling}. The comparison of the quantum resources required by existing methods and our algorithm is summarized in~\cref{tab:compare}.

For estimating $\T(\rho^k)$ to within additive error $\epsilon$, our approach leverages the algorithm from~\cite{quek2024multivariate} as a subroutine. When $k$ exceeds $\widetilde{r} = \min\{r, \lceil\ln(k/\epsilon)\rceil\}$, our method reduces the circuit size of each iteration from $\mathcal{O}(k)$ to $\mathcal{O}(\widetilde{r})$. As a result, it becomes possible to estimate $\T(\rho), \T(\rho^2), \dots, \T(\rho^k)$ for sufficiently large $k$ with a total copy complexity of $\mathcal{O}(k^2/\epsilon^2)$, matching that of prior works~\cite{buhrman2001quantum, ekert2002direct, quek2024multivariate}. However, when estimating $\T(\rho^k)$ for a single value of $k$, our method retains a copy complexity of $\mathcal{O}(k^2/\epsilon^2)$, whereas previous approaches require only $\mathcal{O}(k/\epsilon^2)$ in this case.

The significance of our contribution lies in scenarios where one needs to estimate all moments $\{\T(\rho^i)\}_{i=1}^k$ simultaneously. In such cases, while we maintain the same copy complexity (up to polylogarithmic factors) as existing approaches, our method substantially reduces the quantum circuit resources required for implementation. This leads to a more resource-efficient and scalable procedure, particularly for large $k$, where circuit depth, qubit count, and the number of multi-qubit gates pose practical bottlenecks.

We now present an illustrative example that highlights the utility of the effective rank. Consider a $d$-dimensional quantum state defined by  
\begin{equation}
    \rho = \operatorname{diag}\left(1 - \frac{1}{d}, \frac{1}{d(d - 1)}, \ldots, \frac{1}{d(d - 1)}\right),
\end{equation}
where $\rho$ has rank $d$ and $d = 2^n$. The trace of the $k$-th power of $\rho$ is given by:
\begin{align}
    \T(\rho^k) &= \left(1 - \frac{1}{d}\right)^k + (d - 1) \cdot \left(\frac{1}{d(d - 1)}\right)^k \\
    &= \left(1 - \frac{1}{d}\right)^k + \frac{1}{d^k (d - 1)^{k - 1}}.
\end{align}

Now consider the regime of large $k$, which is the primary focus of our work. For the example state above, setting $k = d$ yields:
\begin{align}
    \T(\rho^k) &= \left(1 - \frac{1}{d}\right)^d + \frac{1}{d^d (d - 1)^{d - 1}} \\
    &\approx \frac{1}{e} + \exp(-\Theta(d \log d)) \approx \frac{1}{e}.
\end{align}
Although one might expect $\T(\rho^k)$ to become negligibly small as $k$ grows large, this example shows that the trace can still retain a significant value, approximately $1/e$, thanks to the contribution of the dominant eigenvalue.

This demonstrates the advantage of the effective rank perspective. Traditional approaches that rely on worst-case rank assumptions would treat this state as full-rank and thus require $\Omega(d) = \Omega(2^n)$ quantum resources to estimate $\T(\rho^k)$ accurately. Here, ``quantum resources'' refer to the number of qubits, the number of multi-qubit gates, and the circuit depth required to implement the estimation algorithm. In contrast, our method based on the effective rank recognizes that only a small subset of eigenvalues contribute meaningfully to the trace, thereby reducing the quantum resource cost to $\mathcal{O}\left(\log\left({d}/{\epsilon}\right)\right) = \mathcal{O}\left(n + \log\left({1}/{\epsilon}\right)\right)$. Crucially, this gain in circuit efficiency is achieved without significantly increasing the number of samples required: the overall sample complexity remains essentially unchanged, up to polylogarithmic factors. We believe that the notion of effective rank can offer similar benefits in many other realistic settings.

\subsection{Trace of quantum state powers with arbitrary observables}\label{chap:obs}
The algorithm we developed for computing the trace of quantum state powers can be extended to address a more generalized problem: estimating $\T(M\rho^k)$, where $M$ represents an arbitrary observable. Successfully estimating this quantity would enable applications in calculating values used as subroutines in virtual distillation~\cite{liang2023unified, huggins2021virtual}, a quantum error mitigation technique.

In this problem, we consider a Pauli decomposition of the observable
\begin{equation}
   M = \sum_{\alpha=1}^{N_M} a_\alpha P_\alpha,
\end{equation}
where $ a_\alpha \in \mathbb{R}$ and
\begin{equation}
   P_\alpha = \sigma_{\alpha_1} \otimes \ldots \otimes \sigma_{\alpha_n} 
\end{equation}
are tensor products of Pauli operators
\begin{equation}
    \sigma_{\alpha_1}, \ldots, \sigma_{\alpha_n} \in \{\sigma_x, \sigma_y, \sigma_z, I\}.
\end{equation}
We assume that the bounded condition
\begin{equation}
   \sum_{\alpha=1}^{N_M} \abs{a_\alpha} = \mathcal{O}(c) 
\end{equation}
holds for some constant $ c $.\\~\\
\underline{{\textbf{[Algorithm 2]} Estimation of $\T(M\rho^k)$}}
\begin{enumerate}
\item Following Steps 1 and 2 of {\textbf{[Algorithm 1]}} in~\cref{chap:algo}, we obtain the values of the elementary symmetric polynomials $b_1, \dots, b_t$.
\item Estimate $ \text{Tr}(M\rho^\ell) $ for $ \ell = 1, 2, \dots, t $ using the method outlined in~\cite{liang2023unified}.\\
\textbf{Note (1):} The quantum circuits required for this step are designed following in~{\cite[Propositions 1 and 2]{liang2023unified}}. As highlighted in their work, the circuit structure depends on the trade-off between qubit-depth and parallelization. In this paper, we focus on describing the high-level procedure without delving into specific implementation details. \\ 
{\textbf{Note (2):} Other methods, such as classical shadows, can be employed to estimate $ \T(M\rho^\ell) $. We emphasize that any method capable of estimating $ \text{Tr}(M\rho^\ell) $ for $ \ell = 1, 2, \dots, t $ can be used as a substitute for this step.}
\begin{enumerate}
    \item For each $\alpha = 1, \dots, N_M$, the following steps are repeated 
    \begin{equation}
        N := \mathcal{O}\left(\frac{\left(\sum_{\alpha=1}^{N_m} \abs{a_\alpha}\right)^2\ln(1/\delta)}{\epsilon^2}  \right)    
    \end{equation}
    times:
    \begin{enumerate}
        \item Prepare a GHZ state and apply a sequence of \textsf{CSWAP} gates.
        \item Apply a controlled-$P_\alpha$ gate to an arbitrary register storing $\rho$.
        \item Repeat the above process and measure the ancillary qubits in the $X$-basis and $Y$-basis, where the $X$-basis measurement is used for the real part estimation and the $Y$-basis measurement is used for the imaginary part estimation. The measurements obtained are then used to estimate $\text{Tr}(P_\alpha \rho^\ell)$ using the similar logic as in Step 1(d) of {\textbf{[Algorithm 1]}}. This estimate, denoted as $\hat{W}_\alpha$, satisfies the following inequality. The value of $\hat{W}_\alpha$ is expressed as the expectation obtained from $N$ repetitions of the measurement process.
        \begin{align}
            & \text{Pr}\left(\abs{\hat{W}_{\alpha} - \text{Tr}(P_\alpha\rho^\ell)}\le \frac{\epsilon}{\sum_{\alpha=1}^{N_M} \abs{a_\alpha}}\right) \nonumber \ge 1-\delta.
        \end{align}
    \end{enumerate}
    \item Finally, the overall expectation value 
    \begin{equation}
    Q_{\ell(\le t),M} = \frac{1}{N_M} \sum_{\alpha=1}^{N_M} a_\alpha\hat{W}_\alpha
    \end{equation}
    serves as an estimate for $\text{Tr}(M\rho^\ell)$. Then this estimate satisfies the inequality below for $\ell = 1, 2, \ldots, t$.
    \begin{align}
        & \text{Pr}\left(\abs{Q_{\ell,M} - \text{Tr}(M\rho^\ell)}\le \epsilon\right) \ge 1-\delta.
    \end{align}
\end{enumerate}
For reference, the sample complexity required in Step 2 is given by:
    \begin{align}
        \mathcal{O}(N_M\cdot N)&=\mathcal{O}\left(\frac{N_M\left(\sum_{\alpha=1}^{N_M} \abs{a_\alpha}\right)^2\ln(1/\delta)}{\epsilon^2}  \right) \\
        &= \mathcal{O}\left(\frac{c^2N_{M}\ln(1/\delta)}{\epsilon^2}\right).
    \end{align}
\item Using $ b_1, \ldots, b_t $ obtained from Step 1 and $ Q_{1,M}, \ldots, Q_{t,M} $ obtained from Step 2, the value of $ \text{Tr}(M\rho^\ell) $ for $ \ell > t $ can be estimated using the following recurrence relation:  
\begin{equation}\label{DefQM}
    Q_{\ell (> t), M} := \sum_{k=1}^{t} (-1)^{k-1} b_k Q_{\ell-k, M} \approx \text{Tr}(M\rho^\ell).
\end{equation}  
\end{enumerate}
Through Step 3, values for $ Q_{t+1, M}, Q_{t+2, M}, \ldots $ can be obtained. In this section, we analyze in detail the conditions on $ t $ required to ensure that the estimated values derived through this process are within an additive error of at most $ \epsilon $. As mentioned, any method capable of estimating $ \text{Tr}(M\rho^\ell) $ for $ \ell = 1, 2, \dots, t $ can be employed in Step 2. The most suitable method should be chosen based on the specific application. For entanglement detection, classical shadows should be used in Step 2, as discussed in~\cref{chap:entangle}. When applying our algorithm for the efficient estimation of $ \T(\rho^k\sigma^l) $, multivariate trace estimation~\cite{quek2024multivariate} is utilized in Step 2, as detailed in~\cref{chap:multipower}.
\begin{theorem}\label{thm:obs}
Suppose that
\begin{equation}
    \varepsilon_{i,M}:=\abs{\epsilon_{i,M}} = \abs{P_{i,M}-Q_{i,M}} < \frac{\epsilon}{4},
\end{equation}
and
\begin{equation}
    \varepsilon_i := \abs{\epsilon_{i}} = \abs{P_{i}-Q_{i}} < \frac{\epsilon}{2\norm{M}_{\infty}kt\ln t},
\end{equation}
holds for $i=1,2,\ldots,t$, 
where the operator norm $\norm{M}_\infty$ is defined corresponding to the $\infty$-norm for vectors $\norm{x}$, as
\begin{equation}
    \norm{M}_\infty = \sup_{x \neq 0} \frac{\norm{Mx}_\infty}{\norm{x}_\infty}.
\end{equation}
{Setting $t = \widetilde{r}_M$ and proceeding with {\textbf{[Algorithm 2]}} based on the recurrence relation~\cref{DefQM}, the following relation always holds:}
\begin{equation}
    \abs{\epsilon_{i,M}} = \abs{P_{i,M}-Q_{i,M}} \leq \epsilon
\end{equation}
for $i=t+1, \ldots, k$. Where $\widetilde{r}_M$ is the effective rank for the observable $M$ defined as:
\begin{equation}\label{eq:eff-rank-observable}
    \widetilde{r}_M = \min\left\{r, \left\lceil\ln\left(\frac{2k\norm{M}_\infty}{\epsilon}\right)\right\rceil\right\}.
\end{equation}
\end{theorem}
\begin{proof}
    See the details in~\cref{thm:obs:proof}.
\end{proof}

Based on~\cref{thm:obs}, the quantum resources required to estimate the trace of quantum state powers with arbitrary observables are derived in~\cref{cor:obs}.
\begin{corollary}\label{cor:obs}
To estimate $\T(M\rho^i)$ for all $i \leq k$ within an additive error of $\epsilon$ and with a success probability of at least $1 - \delta$, where $\delta\in(0,1)$, it is necessary to estimate each $\T(M\rho^j)$ for $j \leq \widetilde{r}_M$ within an additive error of $\varepsilon_{j,M}$ as defined in~\cref{thm:obs}. This can be achieved by using
\begin{equation}
    \mathcal{O}\left(\frac{c^2N_{M}\ln(1/\delta)}{\epsilon^2}\right)
\end{equation}
runs on a constant-depth quantum circuit consisting of $\mathcal{O}(j)$ qubits and $\mathcal{O}(j)$ \textsf{CSWAP} operations, and estimating each $\T(\rho^{j'})$ for $j' \leq \widetilde{r}_M$ within an additive error of $\varepsilon_{j'}$ as defined in~\cref{thm:obs}, by using
\begin{equation}
\widetilde{\mathcal{O}}\left(\frac{k^2\norm{M}_\infty^2\ln(1/\delta)}{\epsilon^2} \right)
\end{equation}
runs on a constant-depth quantum circuit consisting of $\mathcal{O}(j')$ qubits and $\mathcal{O}(j')$ \textsf{CSWAP} operations. Here, the notation $\widetilde{\mathcal{O}}(\cdot)$ hides polylogarithmic factors in $k$.
\end{corollary}
\begin{proof}
    See the details in~\cref{cor:obs:proof}.
\end{proof}
To conclude~\cref{chap:obs}, we summarize our findings in the following proposition:  
\begin{proposition}[Informal, see~\cref{thm:obs} and~\cref{cor:obs}]\label{thm:main-2}  
For the problem of estimating the trace of quantum state powers with arbitrary observables, given a large integer $ k $, it is possible to approximate $ \{\T(M\rho^i)\}_{i=1}^k $ within an additive error of $ \epsilon $ using quantum resources only up to $ \{\T(\rho^i)\}_{i=1}^t $ and $ \{\T(M\rho^i)\}_{i=1}^t $, where $ t $ is given by 
\begin{equation}\label{eq:eff-rank-observable-t}
    t \ge \widetilde{r}_M = \min\left\{r, \left\lceil\ln\left(\frac{2k\norm{M}_\infty}{\epsilon}\right)\right\rceil\right\}.
\end{equation}  
\end{proposition}

The estimation of the trace of quantum state powers with arbitrary observables also applies to the efficient estimation of $\T(\rho^k\sigma^l)$ and is discussed in~\cref{chap:multipower}.

\section{Numerical simulations}\label{chap:simul}
\subsection{Simulation setup}\label{chap:simul-setup}
To validate the findings obtained in~\cref{chap:3}, we conduct numerical simulations to examine the performance of our algorithm. The problem setup to be estimated, including the eigenvalue pattern, is defined as follows and the legend to be used in the graph is shown in~\cref{fig:legend}.
\begin{figure}[htbp!]
  \centering
  \includegraphics[width=0.6\linewidth]{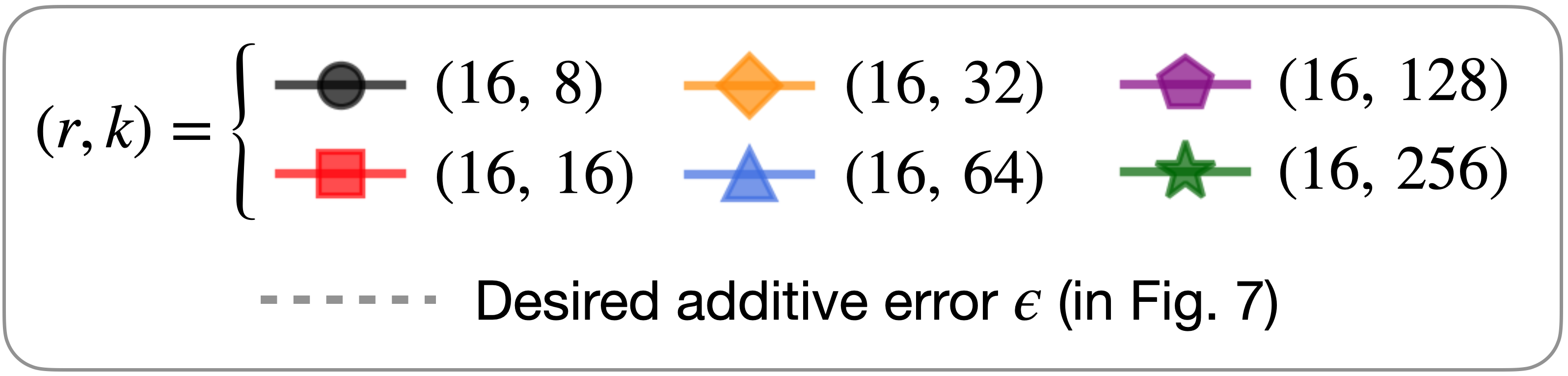}
  \caption{{\textbf{Legends used in the graph.} There are six based on the $(r, k)$ combinations, and in~\cref{fig:dominant-simul}, gray dashed lines are used to further represent the guarantee of estimation within additive error.}}
  \label{fig:legend}
\end{figure}
\begin{itemize}
    \item {Types of eigenvalue distributions:
        \begin{enumerate}[label=(\arabic*)]
            \item \textbf{Geometrically decaying eigenvalues}, $ p_{\max} / p_{\min} = 2^{15}$.
            \item \textbf{Arithmetically decaying eigenvalues}, $p_{\max} - p_{\min} = 0.124$.
            \item \textbf{One dominant eigenvalue} $ p_{\max} \approx 1 $, while the remaining eigenvalues are randomly chosen small values.
            \item \textbf{Identical eigenvalues}, $p_i = {1}/{r}$ for all $i$.
        \end{enumerate}
    \item Rank of the quantum state: $ r = 16 $.
    \item Target power $ k $ for estimating $\text{Tr}(\rho^k)$: $ k \in \{8, 16, 32, 64, 128, 256\} $.
    \item Additive error bound $ \epsilon $ for estimation:  $ \epsilon \in \{ 10^{-1}, 10^{-2}, \dots, 10^{-7} \} $.}
\end{itemize}

In the case of the eigenvalue distribution, the settings for geometrically decaying and arithmetically decaying distributions are mathematically inspired problem setups. For the case of one dominant eigenvalue, the model was formulated under the assumption of an experimental situation where, due to hardware noise or other factors, it is impossible to create a perfect pure state. This situation can be generalized as a scenario where $\widetilde{r} \ll r$ during the operation of our algorithm.

The simulation will be conducted for two different scenarios.  
\begin{enumerate}[label=(\arabic*)]
    \item \textbf{Scenario 1 (simulation of {\textbf{[Algorithm 1]}}):} We evaluate the actual additive error that arises when following the procedure outlined in {\textbf{[Algorithm 1]}} under a given $(r, k, \epsilon)$ setting, using $ t = \widetilde{r} $ for different eigenvalue distributions. (The values of $ \widetilde{r} $ for different $ (k, \epsilon) $ are listed in~\cref{tab:sim-1-1}.) Although Step 1 of {\textbf{[Algorithm 1]}} originally requires a quantum circuit simulation, in our case, we do not employ quantum circuits. Instead, the true value corresponding to $\{\T(\rho^i)\}_{i=1}^t$ is numerically computed, and a simulation is performed using a sampling-based approximation. Specifically, sampling is conducted from a binomial distribution  
    \begin{equation}  
        B\left(n = \left\lceil \left(\frac{k^2}{\epsilon^2}\right) \right\rceil, ~p = \T(\rho^i)\right).  
    \end{equation}  
    To approximate the true value, $n$ independent random variables are drawn from this binomial distribution, and their empirical mean is used as the estimate. Since $n$ is chosen to satisfy~\cref{cor:eff}, the estimation error can be maintained below $\epsilon$.
    \item \textbf{Scenario 2 (simulation of~\cref{lem:eff}):} We investigate how the error evolves as the value of $ t $ is varied. In particular, we examine the error trend when $ t < r $ or even when $ t < \widetilde{r} $. The objective is to determine the minimum value of $ t $ required to ensure that the estimation remains within a sufficiently small additive error across various distributions. In this scenario, $k$ is fixed at 32.
\end{enumerate}

\begin{table}[]
\centering
\resizebox{0.6\columnwidth}{!}{%
\begin{tabular}{c|ccccccc}
$(k,\epsilon)$      & $10^{-1}$ & $10^{-2}$ & $10^{-3}$ & $10^{-4}$ & $10^{-5}$ & $10^{-6}$ & $10^{-7}$ \\ \hline
$8$   & 6         & 8         & 10        & 12        & 15        & 16        & 16        \\
$16$  & 6         & 9         & 11        & 13        & 15        & 16        & 16        \\
$32$  & 6         & 9         & 11        & 13        & 15        & 16        & 16        \\
$64$  & 8         & 10        & 12        & 15        & 16        & 16        & 16        \\
$128$ & 8         & 11        & 13        & 15        & 16        & 16        & 16        \\
$256$ & 9         & 11        & 14        & 16        & 16        & 16        & 16       
\end{tabular}%
}
\caption{{\textbf{The value of $\widetilde{r}$ as a function of $(k, \epsilon)$.} The value of $ t $ used in Scenario 1 is $ \widetilde{r} = \min\left\{r, \left\lceil\ln\left({2k}/{\epsilon}\right)\right\rceil\right\} $. Note that $ r = 16 $.}}
\label{tab:sim-1-1}
\end{table}

\subsection{Simulation result}\label{chap:simul-result}
The simulation results for geometrically decaying, arithmetically decaying, one dominant, and identical eigenvalues are shown in~\cref{fig:geo-simul},~\cref{fig:arith-simul},~\cref{fig:dominant-simul}, and~\cref{fig:id-simul}, respectively. Each figure consists of three subfigures: (a) the distribution of the eigenvalues, (b) Scenario 1—simulation of {\textbf{[Algorithm 1]}}, and (c) Scenario 2—simulation of~\cref{lem:eff}.

For every eigenvalue distribution, the experimental error in Scenario 1 is smaller than the target additive error $\epsilon$, which strengthens the credibility of {\textbf{[Algorithm 1]}}. In the cases of geometrically decaying, arithmetically decaying, and identical eigenvalues, the discrepancy between the target error and the experimental error is quite large. The case of one dominant eigenvalue gives the tightest result.

For every eigenvalue distribution, the experimental error in Scenario 2 is also smaller than the target additive error $\epsilon$, further enhancing the credibility of~\cref{lem:eff}. Only $\{\T(\rho^i)\}_{i=1}^t$ is obtained from quantum resources, while $\T(\rho^{t+1})$ to $\T(\rho^k)$ are computed using the recurrence relation described in the algorithm. As mentioned earlier, our simulation uses a sampling-based approximation instead of quantum resources. The graph presents both
\begin{equation}
    \max_{j \in \{t+1, \dots, k\}} \abs{P_j - \widetilde{P}_j}
\end{equation}
and the theoretical bound we derived, ${t}/{k!}$. In the cases of geometrically decaying, arithmetically decaying eigenvalues, and one dominant eigenvalue, the discrepancy between the theoretical bound ${k}/{t!}$ and the experimental error is quite large. The case of identical eigenvalues gives the tightest result. For every distribution we simulated, $t=8$ is sufficient to keep the experimental error below a low threshold (e.g., always smaller than $10^{-6}$, which is sufficiently small).

Additionally, as the power $k$ increases, both the scale of $\mathrm{Tr}(\rho^k)$ and the scale of the absolute error become very small, sometimes even dropping below the machine epsilon, which represents the smallest numerical difference a computer can accurately represent in floating-point arithmetic. To eliminate errors caused by floating-point precision limitations, we implemented our algorithm using integer fractions instead of floating-point types for iterative estimations. Specifically, we used Python’s built-in \texttt{fractions.Fraction} class, which represents rational numbers exactly as ratios of two integers. This allowed us to perform arithmetic operations with full precision, avoiding the accumulation of rounding errors that typically arise in floating-point computations. Thanks to this exact representation, our simulation was able to detect discrepancies as small as on the order of $10^{-200}$. 

We note that the use of exact rational arithmetic in our work is primarily for methodological purposes. In realistic near-term quantum experiments, errors from finite measurement statistics, decoherence, and other hardware imperfections are expected to far exceed floating-point precision limits. Nevertheless, we performed sampling-based simulations so that, even though such physical noise sources were not modeled, the results remain statistically meaningful. This setup aligns with the aim of our study, which is to evaluate and validate the intrinsic performance of the algorithm under idealized, noise-free conditions. Such extremely low error levels should not be expected in practice on current quantum devices.

Here, we uncover a new insight: in~\cref{proof:approx}, the theoretical bound is derived using the scaling difference between factorial and exponential functions, such as $t! \geq 2^t$. However, this approach may not provide a sufficiently tight bound. Obtaining a closed-form lower bound for $t$ analytically is extremely challenging, but considering Stirling’s approximation,  
\begin{equation}
    n! \sim \sqrt{2\pi n} \left(\frac{n}{e}\right)^n \left(1 + \frac{1}{12n} + \frac{1}{288n^2} + \cdots \right),    
\end{equation}
we observe that the lower bound for $t$ could be as low as  
\begin{equation}
    \mathcal{O}\left(\frac{\ln\left({k}/{\epsilon}\right)}{\ln \ln \left({k}/{\epsilon}\right)}\right),    
\end{equation}
suggesting a potentially looser bound than initially expected. And the simulation results based on this bound are included in~\cref{chap:add-sim}.

\clearpage  % 현재 페이지 내용을 마무리하고 새 페이지 시작
\vspace*{0.15\textheight}

\noindent
\begin{minipage}{0.45\textwidth}
\begin{figure}[H]
    \centering
    \begin{subfigure}{\linewidth}
        \centering
        \caption{Distribution: Geometrically decaying}
        \includegraphics[width=\linewidth]{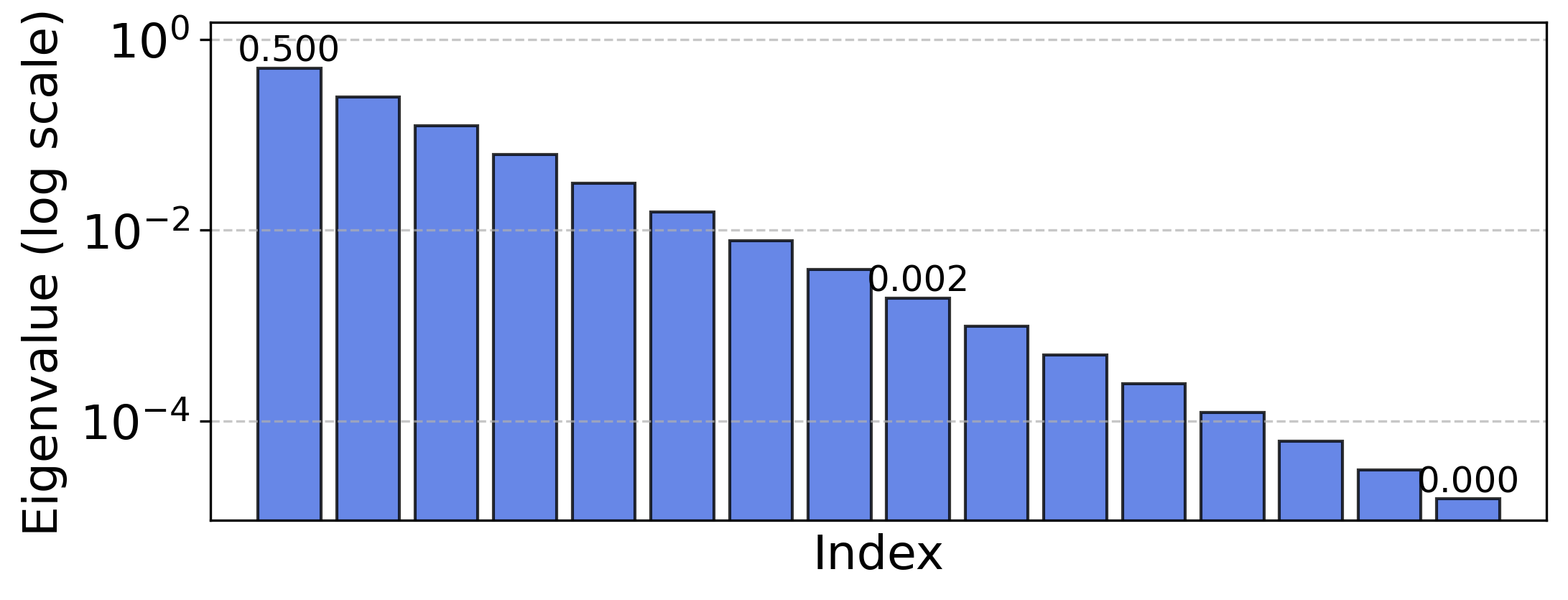}
        \label{fig:geo-dist}
    \end{subfigure}
    \vfill
    \begin{subfigure}{\linewidth}
        \centering
        \caption{Scenario 1: Geometrically decaying}
        \includegraphics[width=\linewidth]{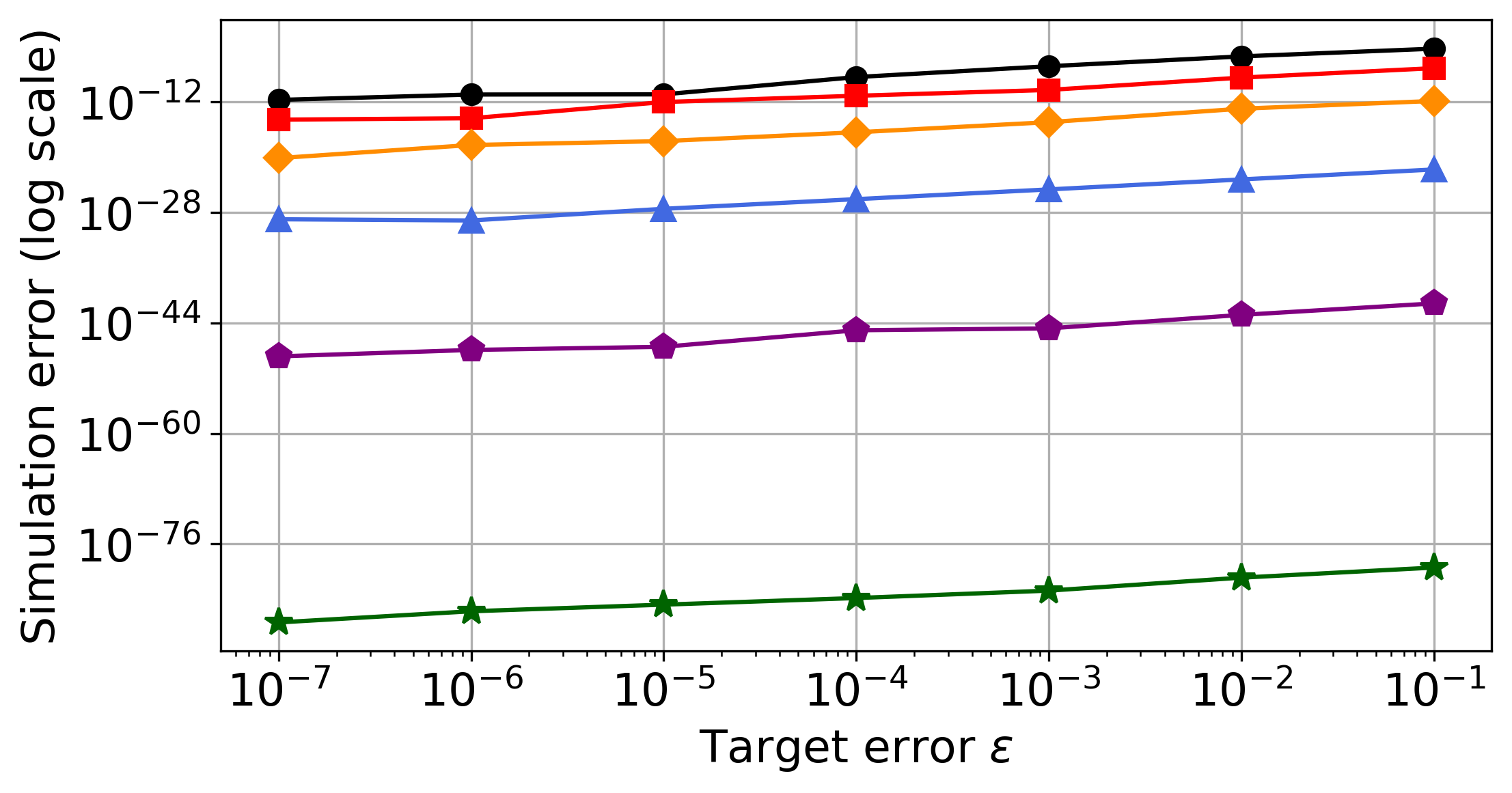}
        \label{fig:geo-1}
    \end{subfigure}
    \vfill
    \begin{subfigure}{\linewidth}
        \centering
        \caption{Scenario 2: Geometrically decaying}
        \includegraphics[width=\linewidth]{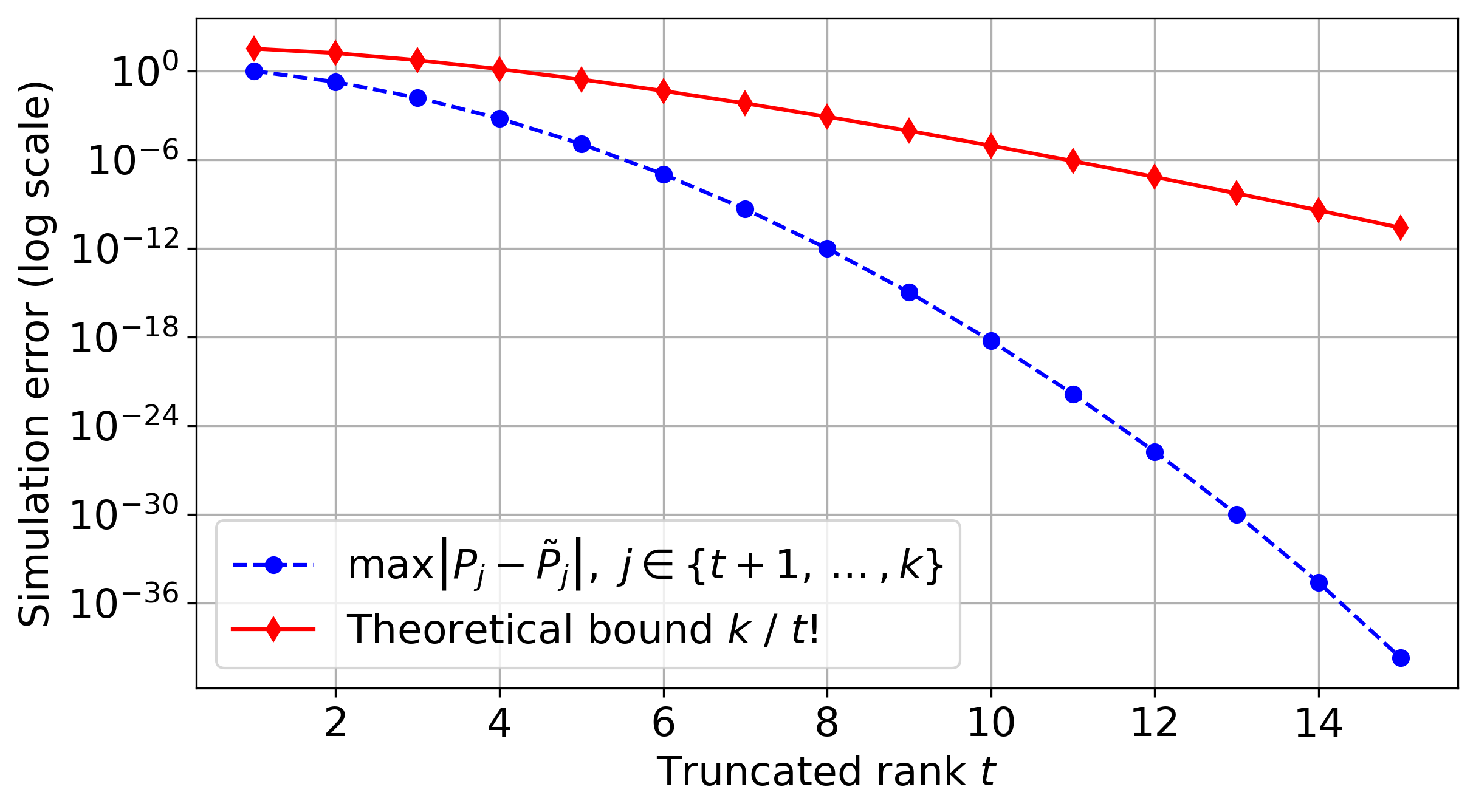}
        \label{fig:geo-2}
    \end{subfigure}
    \caption{Simulation results for geometrically decaying eigenvalues.}
    \label{fig:geo-simul}
\end{figure}
\end{minipage}
\hfill
\begin{minipage}{0.45\textwidth}
\begin{figure}[H]
    \centering
    \begin{subfigure}{\linewidth}
        \centering
        \caption{Distribution: Arithmetically decaying}
        \includegraphics[width=\linewidth]{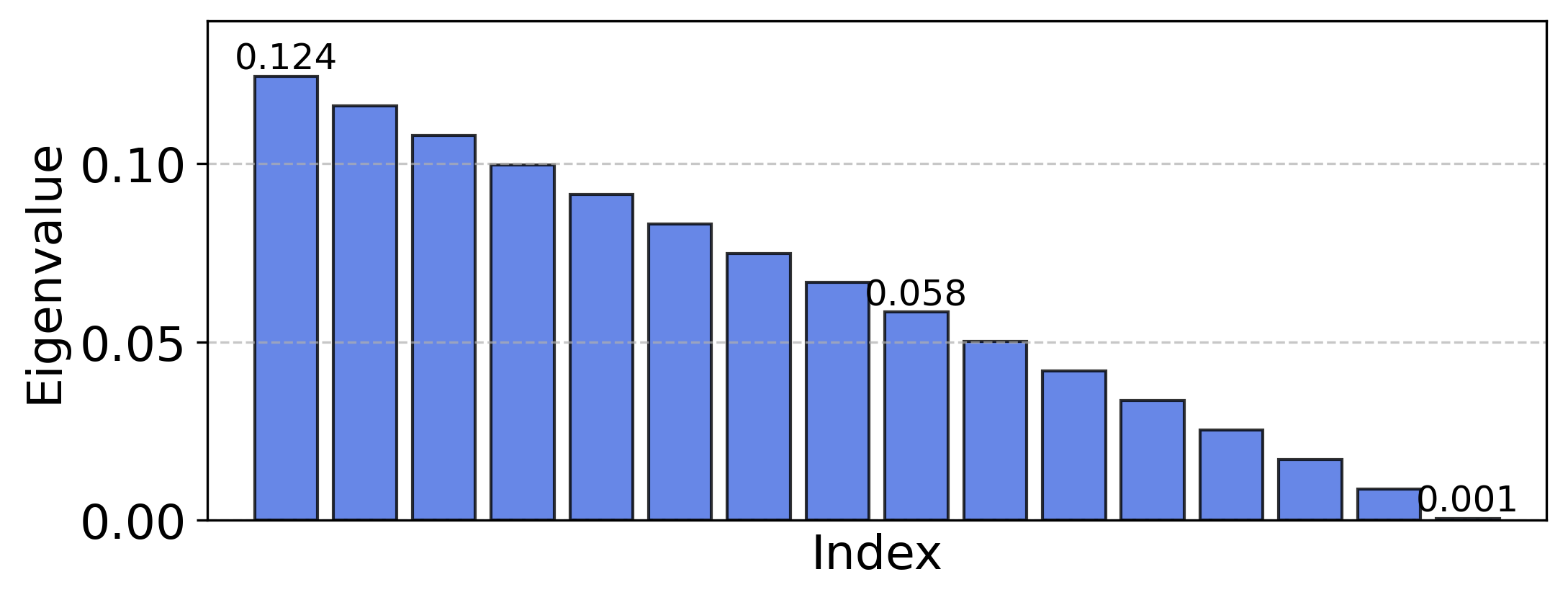}
        \label{fig:arith-dist}
    \end{subfigure}
    \vfill
    \begin{subfigure}{\linewidth}
        \centering
        \caption{Scenario 1: Arithmetically decaying}
        \includegraphics[width=\linewidth]{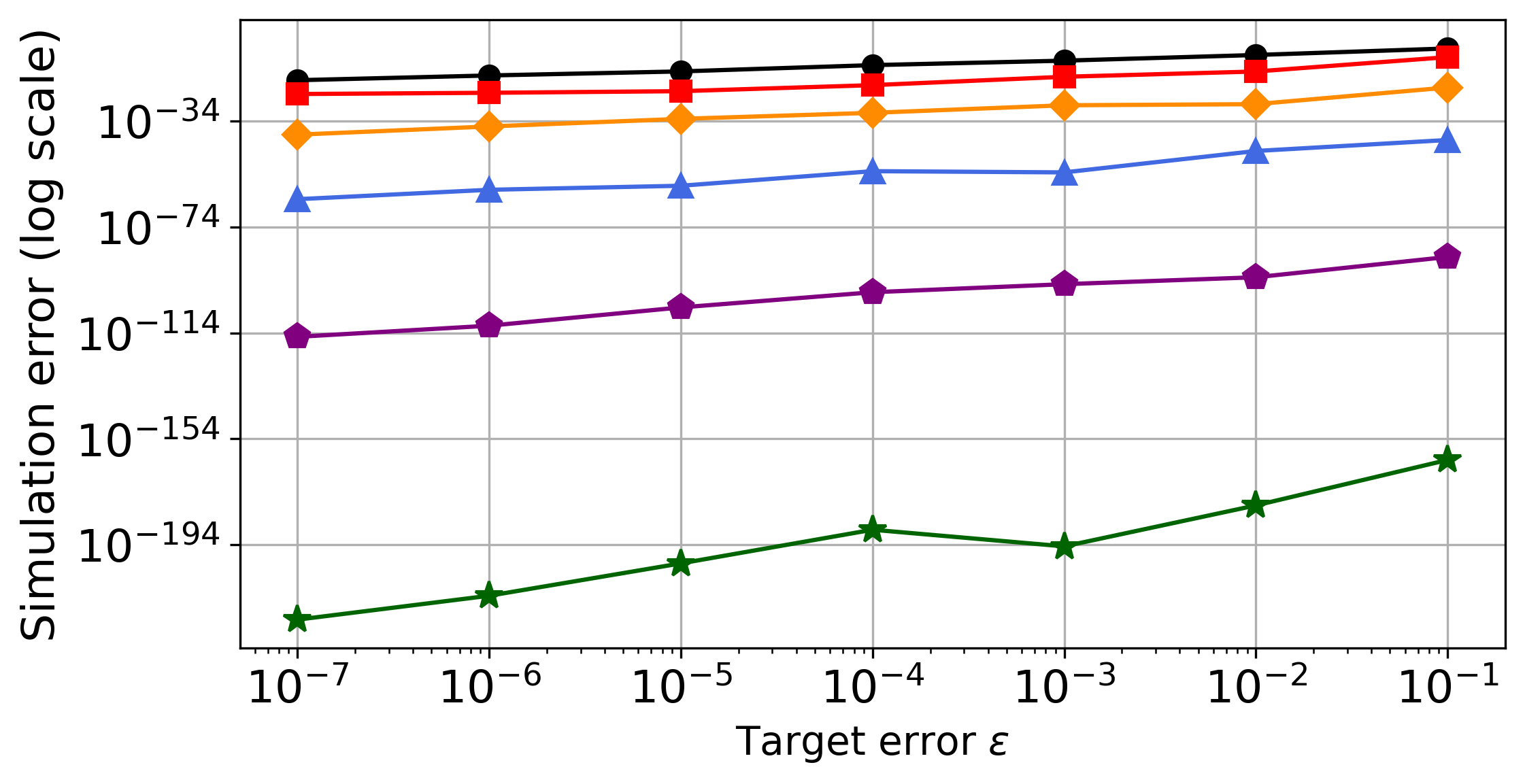}
        \label{fig:arith-1}
    \end{subfigure}
    \vfill
    \begin{subfigure}{\linewidth}
        \centering
        \caption{Scenario 2: Arithmetically decaying}
        \includegraphics[width=\linewidth]{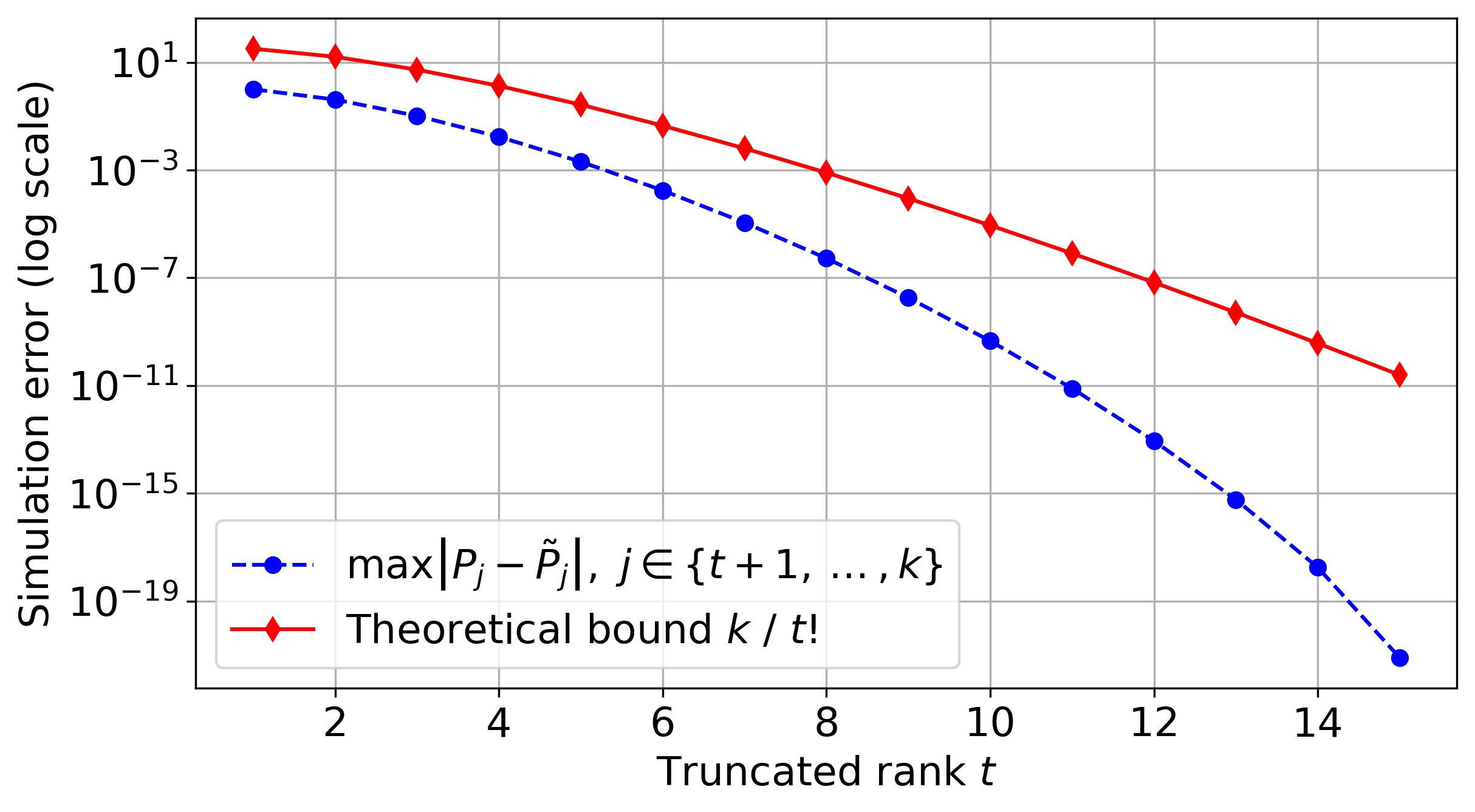}
        \label{fig:arith-2}
    \end{subfigure}
    \caption{Simulation results for arithmetically decaying eigenvalues.}
    \label{fig:arith-simul}
\end{figure}
\end{minipage}

\vspace*{0.15\textheight}

\clearpage  % 현재 페이지 내용을 마무리하고 새 페이지 시작
\vspace*{0.15\textheight}
\noindent

\begin{minipage}{0.45\textwidth}
\begin{figure}[H]
    \centering
    \begin{subfigure}{\linewidth}
        \centering
        \caption{Distribution: One dominant}
        \includegraphics[width=\linewidth]{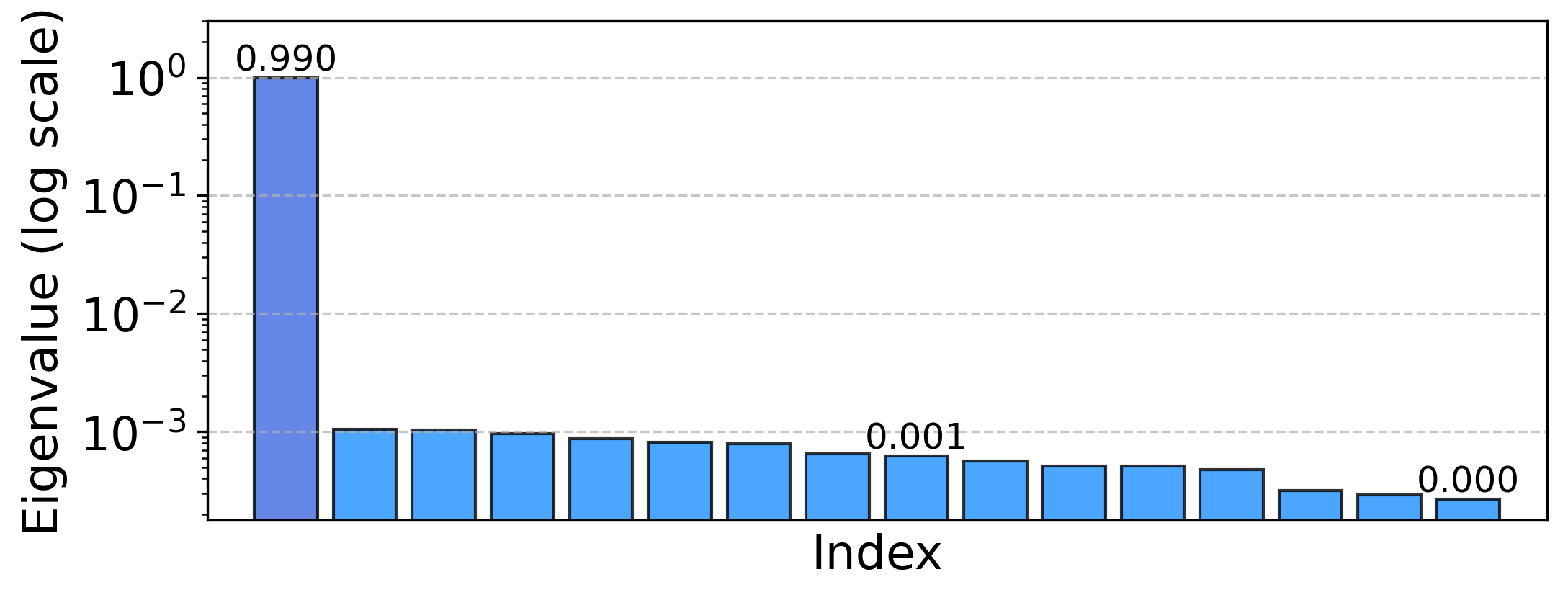}
        \label{fig:dominant-dist}
    \end{subfigure}
    \vfill
    \begin{subfigure}{\linewidth}
        \centering
        \caption{Scenario 1: One dominant}
        \includegraphics[width=\linewidth]{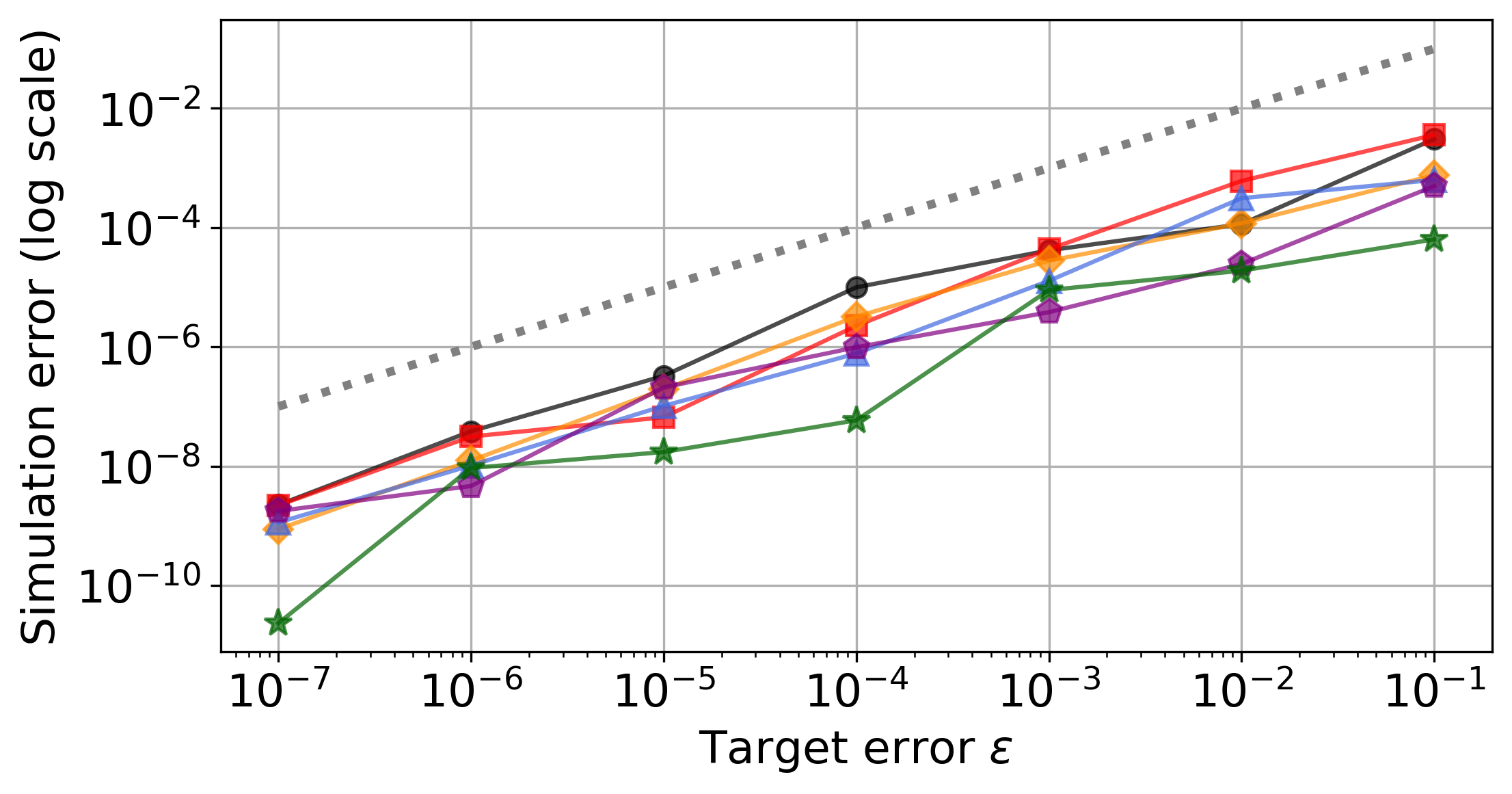}
        \label{fig:dominant-1}
    \end{subfigure}
    \vfill
    \begin{subfigure}{\linewidth}
        \centering
        \caption{Scenario 2: One dominant}
        \includegraphics[width=\linewidth]{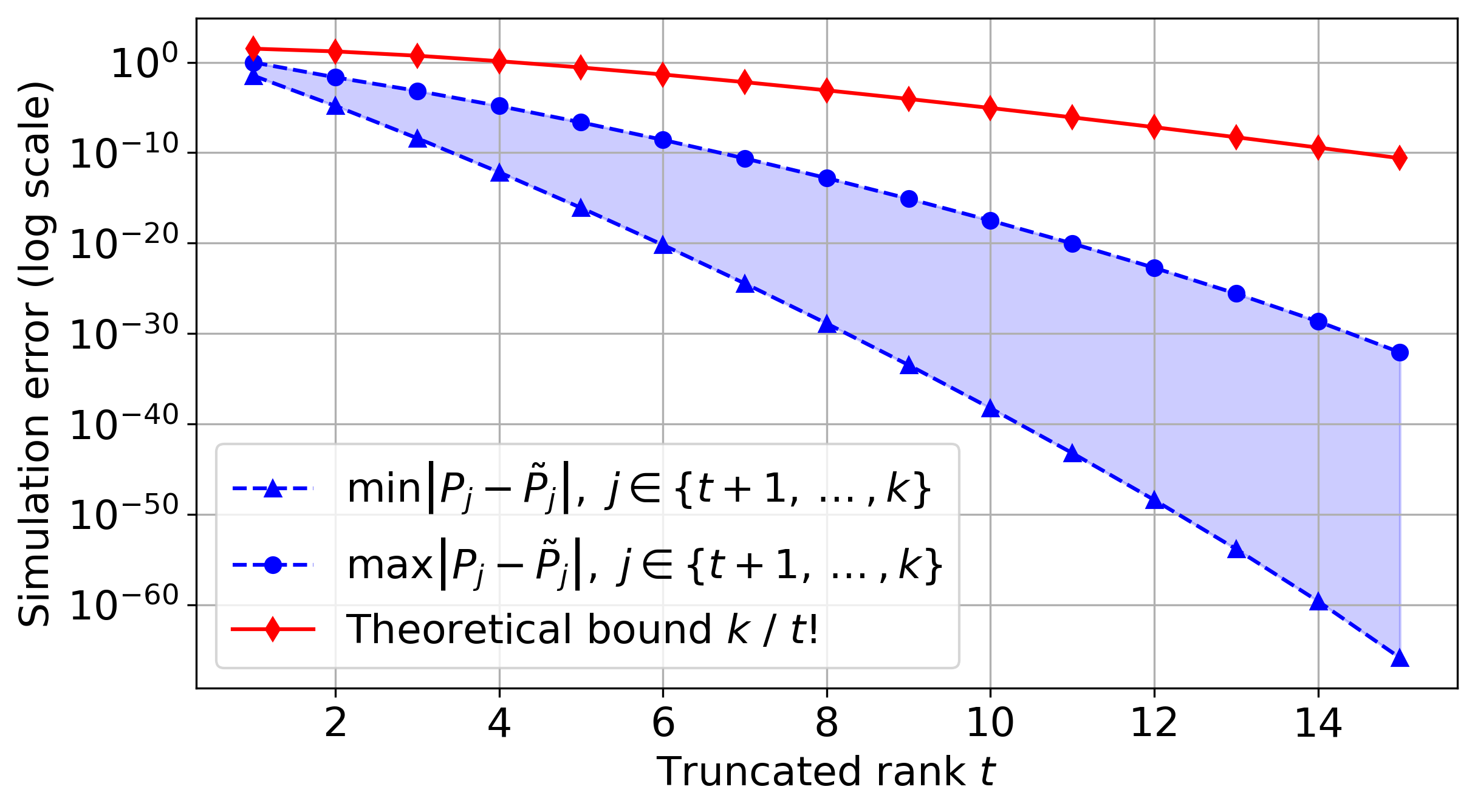}
        \label{fig:dominant-2}
    \end{subfigure}
    \caption{Simulation results for one dominant eigenvalue.}
    \label{fig:dominant-simul}
\end{figure}
\end{minipage}
\hfill
\begin{minipage}{0.45\textwidth}
\begin{figure}[H]
    \centering
    \begin{subfigure}{\linewidth}
        \centering
        \caption{Distribution: Identical}
        \includegraphics[width=\linewidth]{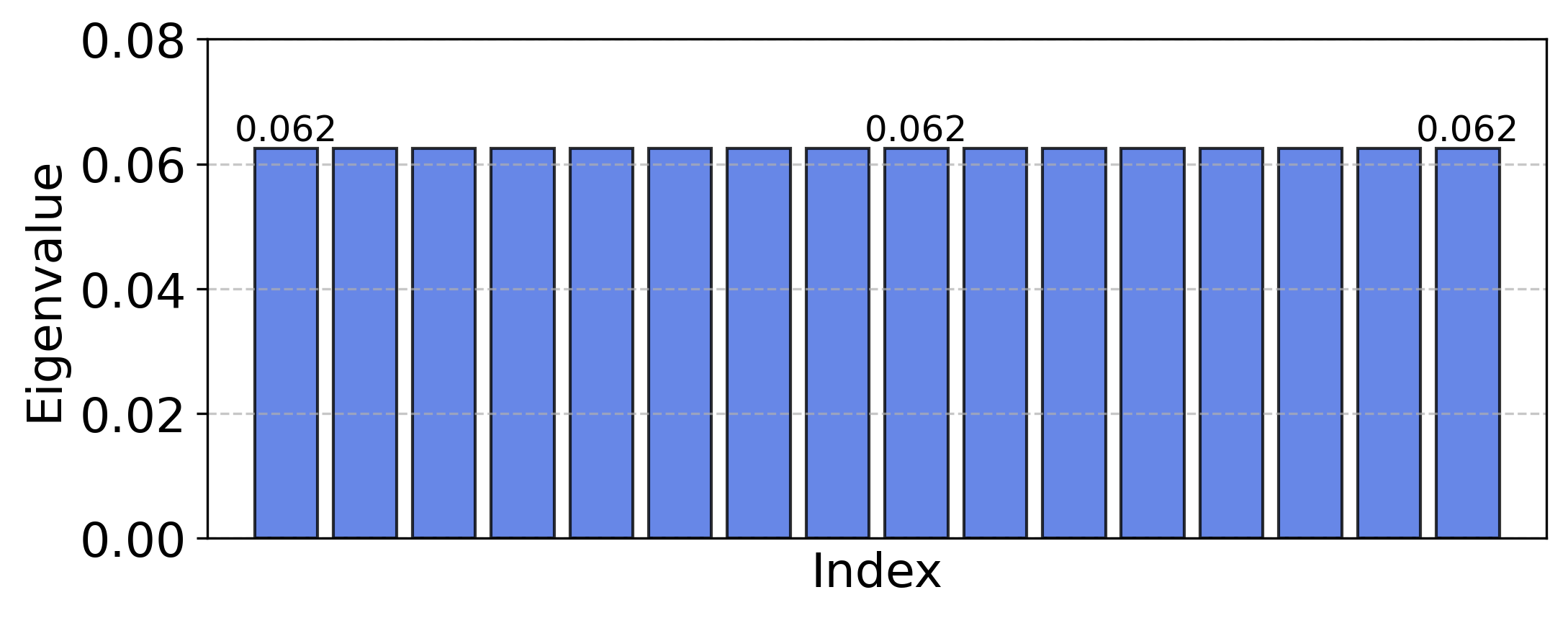}
        \label{fig:id-dist}
    \end{subfigure}
    \vfill
    \begin{subfigure}{\linewidth}
        \centering
        \caption{Scenario 1: Identical}
        \includegraphics[width=\linewidth]{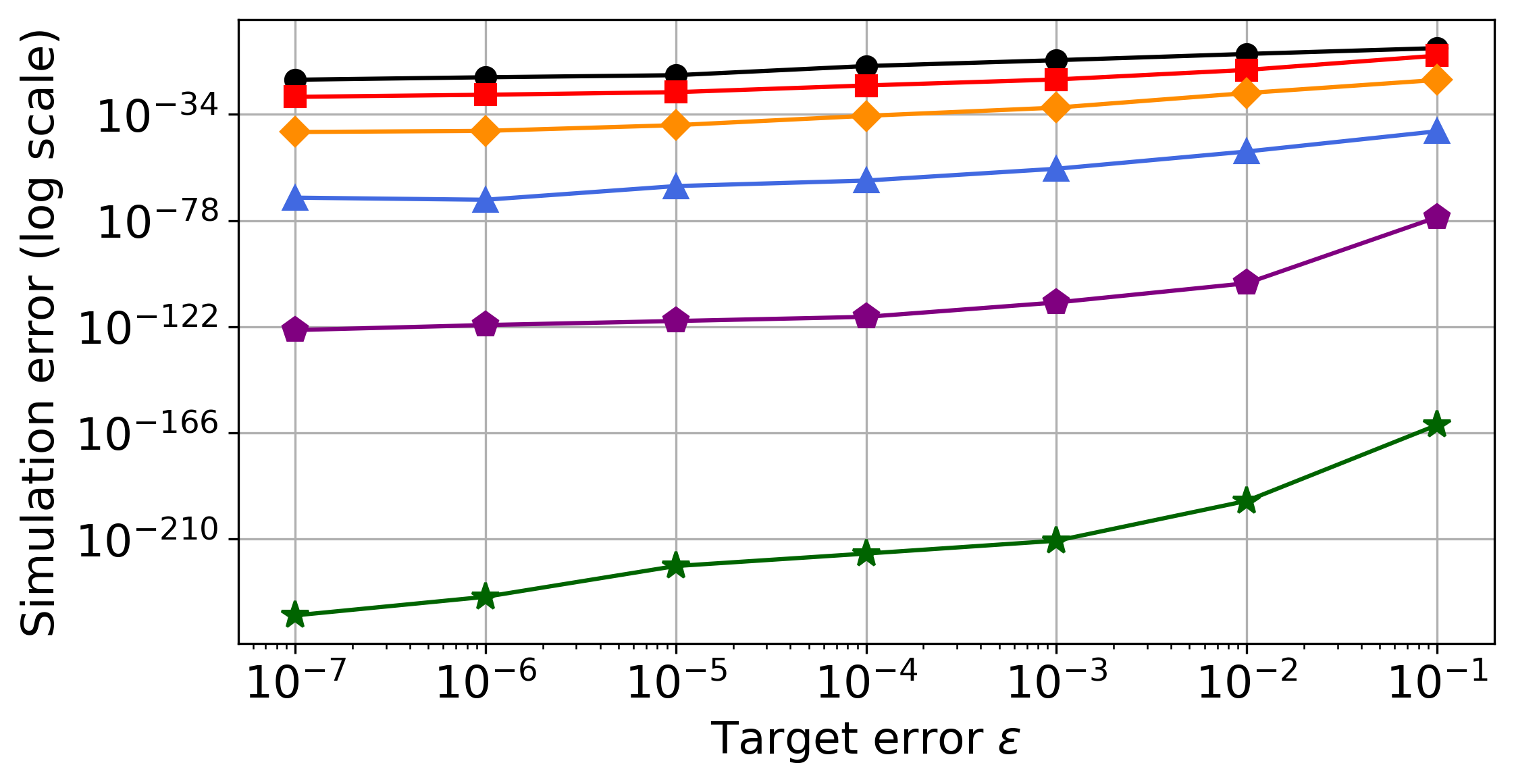}
        \label{fig:id-1}
    \end{subfigure}
    \vfill
    \begin{subfigure}{\linewidth}
        \centering
        \caption{Scenario 2: Identical}
        \includegraphics[width=\linewidth]{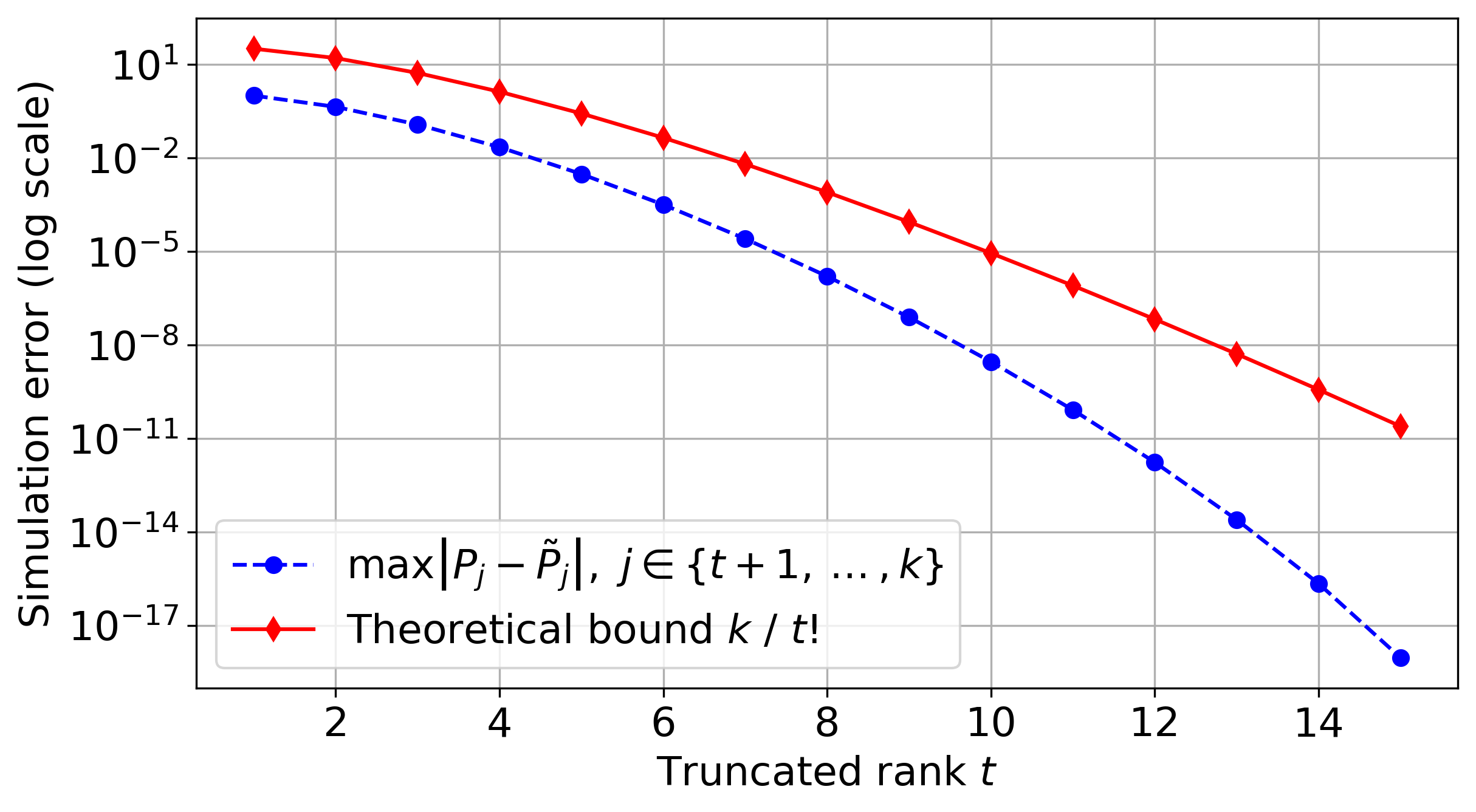}
        \label{fig:id-2}
    \end{subfigure}
    \caption{Simulation results for identical eigenvalues.}
    \label{fig:id-simul}
\end{figure}
\end{minipage}
\vspace*{0.15\textheight}

\newpage
\section{Applications in quantum information}\label{chap:application}
Now, let's explore several use cases of how our \textit{rank is all you need} and \textit{effective rank is all you need} ideas can be efficiently applied to quantum information processing tasks.

\subsection{Nonlinear function calculations for quantum states}
Applying~\cref{cor:rank} to~{\cite[Theorem 5]{quek2024multivariate}}, we can enhance the theorem.

\begin{theorem}\label{thm:nonlinear}
Let $\rho$ be a quantum state with rank $r$. Suppose there exist $\epsilon > 0$ and a slowly-growing function $C$ (as a function of $k$) such that $g: \mathbb{R} \rightarrow \mathbb{R}$ is approximated by a degree $k$ polynomial
\begin{equation}
    f(x) = \sum_{i=0}^k c_i x^i
\end{equation}
on the interval $[0,1]$, in the sense that
\begin{equation}
    \sup_{x\in [0,1]} \abs{g(x)-f(x)} < \frac{\epsilon}{2r},
\end{equation}
and
\begin{equation}
    \sum_{i=0}^k \abs{c_i} < C.
\end{equation}
Then estimating $\T(g(\rho))$ within an $\epsilon$ additive error and with a success probability of at least $1 - \delta~$, where $\delta\in(0, 1)$ requires
\begin{equation}
    \mathcal{O}\left(\frac{C^2k^2 \widetilde{r}^4\ln^2 \widetilde{r}\ln(1/\delta)}{\epsilon^2} \right) = \widetilde{\mathcal{O}}\left(\frac{C^2k^2\ln(1/\delta)}{\epsilon^2}\right)
\end{equation}
copies of $\rho$ and
\begin{equation}
    \mathcal{O}\left(\frac{C^2k^2 \widetilde{r}^3\ln^2 \widetilde{r}\ln(1/\delta)}{\epsilon^2} \right) = \widetilde{\mathcal{O}}\left(\frac{C^2k^2\ln(1/\delta)}{\epsilon^2}\right)
\end{equation}
runs on a constant-depth quantum circuit consisting of {$\mathcal{O}(\widetilde{r})$} qubits and {$\mathcal{O}(\widetilde{r})$} \textsf{CSWAP} operations. Here, the notation $\widetilde{\mathcal{O}}(\cdot)$ hides polylogarithmic factors in $k$ and $1/\epsilon$.
\end{theorem}
In the original theorem mentioned in \cite{quek2024multivariate}, 
\begin{equation}
    \mathcal{O}\left(\frac{C^2k^2\ln(1/\delta)}{\epsilon^2}\right)    
\end{equation}
copies of $\rho$ were required, and the circuit consisted of $\mathcal{O}(k)$ qubits and $\mathcal{O}(k)$ \textsf{CSWAP} operations. 

According to~\cref{thm:rank} and~\ref{thm:eff}, it suffices to estimate $\{\T(\rho^i)\}_{i=1}^{\tilde{r}}$ with additive error at most $\epsilon / (2Ck\widetilde{r} \ln \widetilde{r})$ in order to approximate $\{\T(\rho^i)\}_{i=1}^k$ up to an additive error of $\epsilon / C$. Following the same reasoning as in~\cref{eq:runs} and~\cref{eq:runs-cal}, we can determine both the circuit repetition count and the required number of copies of $\rho$ for this estimation.

By applying the result of~\cref{cor:rank}, the total number of copies of $\rho$ required to achieve an $\epsilon$-additive estimation with failure probability at most $\delta$ is given by
\begin{equation} 
    \mathcal{O}\left(\frac{C^2 k^2 \widetilde{r}^4 \ln^2 \widetilde{r}\ln(1/\delta)}{\epsilon^2} \right).
\end{equation}
The corresponding algorithm can be implemented on a constant-depth quantum circuit that acts on $\mathcal{O}(\widetilde{r})$ qubits and uses $\mathcal{O}(\widetilde{r})$ \textsf{CSWAP} gates.

Typically, $k$ is much larger than $\widetilde{r}$, so our enhanced theorem offers advantages for estimating $g(\rho)$. When $g(x) = e^{\beta x}$, $C$ becomes $e^{\abs{\beta}}$. We can efficiently estimate $\T(e^{\beta \rho})$ using~\cref{thm:nonlinear}, which has applications in thermodynamics and the density exponentiation algorithm~\cite{lloyd2014quantum, kimmel2017hamiltonian, quek2024multivariate}.

\subsection{Quantum Gibbs state preparation}
We highlight that our method improves the efficiency of preparing the quantum Gibbs state. The truncated Taylor series:
\begin{equation}
    S_q(\rho) = \sum_{i=1}^q \T\left(\left(\rho-I\right)^q\rho\right)
\end{equation}
is used as the cost function for variational quantum Gibbs state preparation~\cite{wang2021variational}. It is shown that the fidelity $F(\rho(\theta_0), \rho_G)$ between the optimized state $\rho(\theta_0)$ and the Gibbs state $\rho_G$ is bounded by
\begin{equation}
    F(\rho(\theta_0), \rho_G) \geq 1-\sqrt{2\left(\beta\epsilon + \frac{2r}{q+1}(1-\Delta)^{q+1}\right)},
\end{equation}
where $\beta$ is the inverse temperature of the system, and $\Delta$ is a constant that satisfies
\begin{equation}
    -\Delta \ln(\Delta) < \frac{1}{q+1}(1-\Delta)^{q+1}.
\end{equation}
By using the inequality
\begin{equation}
    D(\rho(\theta_0), \rho_G) < \sqrt{1-F\left(\rho\left(\theta_0\right), \rho_G\right)},
\end{equation}
to achieve $D\left(\rho(\theta_0), \rho_G\right) < \epsilon$, we need to set $q = \mathcal{O}\left({r}/{\epsilon^{4}}\right)$, where $D$ is the trace distance. Using previous methods, $q = \mathcal{O}\left({r}/{\epsilon^{4}}\right)$ qubits and \textsf{CSWAP} operations are required, which are impractical for near-term quantum devices. Our work significantly reduces the number of qubits and \textsf{CSWAP} operations to {$\mathcal{O}(\widetilde{r})$, exponentially reducing the quantum resources.} This demonstrates that our method makes the preparation of the quantum Gibbs state using the truncated Taylor series much more feasible.

\subsection{Efficient estimation of the trace of products of quantum state powers}\label{chap:multipower}
In this section, we discuss the efficient estimation of the set
\begin{equation}
    \{\T(\rho^i\sigma^j):(i,j)\in[k]\times[l]\}.    
\end{equation}
The core routine relies on {\textbf{[Algorithm 2]}} from~\cref{chap:obs}. Extending the problem from the general estimation of the trace of quantum state powers to this broader setting is crucial, as it enables the estimation of various distance measures, making this generalization highly significant.  
\begin{theorem}\label{thm:multi}  
For the problem of estimating the trace of products of quantum state powers, given large integers $k, l$, it is possible to approximate $\{\T(\rho^i\sigma^j)\}_{(i,j)=(1,1)}^{(k,l)}$ within an additive error of $\epsilon$ using quantum resources only up to $ \{\T(\rho^i)\}_{i=1}^t $, $ \{\T(\sigma^i)\}_{i=1}^t $, and $\{\T(\rho^i\sigma^j)\}_{(i,j)=(1,1)}^{(t,t)}$, where $ t $ satisfies  
\begin{equation}\label{eq:eff-rank-multi}  
    t \geq \widetilde{R} = \min\left\{r, \left\lceil\ln\left(\frac{4k + 4l}{\epsilon}\right)\right\rceil\right\}.  
\end{equation}  
We define $ \widetilde{R} $ as the effective rank of the quantum states $ \rho $ and $ \sigma$, where
\begin{equation}
    r = \max\{\textnormal{rank}(\rho),~\textnormal{rank}(\sigma)\}.
\end{equation}
To clarify the notation, $ \{\T(\rho^i\sigma^j)\}_{(i,j)=(1,1)}^{(a,b)} $ refers to the set of values obtained by calculating $ \T(\rho^i \sigma^j) $ for all $ (i,j) \in [a] \times [b] $ where $[n]$ denotes the set of natural numbers from 1 to $n$.
\end{theorem}
\begin{proof}
    See the details in~\cref{thm:multi:proof}.
\end{proof}
To aid in understanding the proof of~\cref{thm:multi}, we can summarize the idea sketch of the proof in the following manner, corresponding to {\textbf{[Algorithm 3]}}.\\~\\
\underline{{\textbf{[Algorithm 3]} Estimation of $\T(\rho^k\sigma^l)$}}
\begin{enumerate}
    \item For a fixed index $i$, we begin by estimating the sets $ \{\T(\sigma^i)\}_{i=1}^t $ and $ \{\T(\rho^i\sigma^j)\}_{(i,j)=(1,1)}^{(t,t)} $, which enables the application of {\textbf{[Algorithm 2]}} with $M = \rho^i$. This allows us to compute the estimated values for $ \{\T(\rho^i\sigma^j)\}_{j=1}^l $ for the fixed $i$.
    \item Repeating Step 1 for $i = 1, 2, \dots, t$, we obtain the values for $ \{\T(\rho^i\sigma^j)\}_{(i,j)=(1,1)}^{(t,l)} $.
    \item For a fixed index $j$, we then estimate the sets $ \{\T(\rho^i)\}_{i=1}^t $ and $ \{\T(\rho^i\sigma^j)\}_{i=1}^t $, allowing us to apply {\textbf{[Algorithm 2]}} once more, this time setting $M = \sigma^j$. This step computes the estimated values for $ \{\T(\rho^i\sigma^j)\}_{i=1}^k $.
    \item With the values for $ \{\T(\rho^i\sigma^j)\}_{(i,j)=(1,1)}^{(t,l)} $ already obtained in Step 2, the final process yields the estimated values for $ \{\T(\rho^i\sigma^j)\}_{(i,j)=(1,1)}^{(k,l)} $.
\end{enumerate}
The estimation of the values of $ \{\T(\rho^i)\}_{i=1}^t $, $ \{\T(\sigma^i)\}_{i=1}^t $, and $ \{\T(\rho^i\sigma^j)\}_{(i,j)=(1,1)}^{(t,t)} $ requires the same procedure as Step 2 of {\textbf{[Algorithm 2]}}. Consequently, the quantum resources required for this process can be estimated as follows: for the estimation of $ \{\T(\rho^i\sigma^j)\}_{(i,j)=1}^{(k,l)} $, at most $ \mathcal{O}(t) $ qubits and $ \mathcal{O}(t) $ \textsf{CSWAP} gates are necessary. This improves upon the previous result in \cite{quek2024multivariate}, which required $ \mathcal{O}(k+l) $ qubits and $ \mathcal{O}(k+l) $ \textsf{CSWAP} gates. Furthermore, in terms of copy complexity, each quantum state $\rho$ and $\sigma$ is required $\widetilde{\mathcal{O}}\left({k^2}/{\epsilon^2}\right)$ and $\widetilde{\mathcal{O}}\left({l^2}/{\epsilon^2}\right)$ times, respectively (details in~\cref{thm:multi:proof}). Here, the notation $\widetilde{\mathcal{O}}(\cdot)$ hides polylogarithmic factors in $k$ and $l$. This improves upon previous studies, where $\rho$ and $\sigma$ were required $\widetilde{\mathcal{O}}\left({k^2l}/{\epsilon^2}\right)$ and $\widetilde{\mathcal{O}}\left({kl^2}/{\epsilon^2}\right)$ times, respectively. An efficient algorithm for estimating $\{\T(\rho^i\sigma^j)\}_{(i,j)=1}^{(k,l)}$ can be widely applied to various quantum information tasks. For example, it can be used to compute the Schatten-$p$ distance, defined as
\begin{equation}\label{eq:schatten}
\norm{\rho-\sigma}_p = \left(\T\left[\abs{\rho-\sigma}^p\right]\right)^{\frac{1}{p}}.
\end{equation}
It also applies to the estimation of other distance measures, such as
\begin{equation}
K_{\alpha}(\rho, \sigma) = \T(\left(1+\rho\right)^\alpha\left(1+\sigma\right)^{1-\alpha}),
\end{equation}
which satisfies faithfulness and the data processing inequality under unital quantum channels~\cite{quek2024multivariate}.

\subsection{Entanglement detection}\label{chap:entangle}
Determining whether a quantum state is separable or entangled is a fundamental problem in quantum information theory. It is well known that a separable quantum state $\rho_{AB}$ always has a positive semi-definite (PSD) partial transpose (PT), denoted as $\rho_{AB}^{\Gamma_B}$. By contraposition, if $\rho_{AB}^{\Gamma_B}$ has a negative eigenvalue, then $\rho_{AB}$ must be entangled. For brevity, we denote the partial transpose of $\rho$ as $\rho^\Gamma$. The $k$-th PT moment is defined as
\begin{equation}
    p_k^\text{PT} = \T((\rho^\Gamma)^k).
\end{equation}
PT moments are typically estimated using classical shadows~\cite{neven2021symmetry, elben2020mixed}. By leveraging PT moments and the Newton-Girard method, the presence of a negative eigenvalue can be detected. The PT moments required for entanglement detection are $p_1^\text{PT}, p_2^\text{PT}, \dots, p_r^\text{PT}$, where $r$ is the rank of {$\rho^\Gamma$. Let $\lambda_1, \ldots, \lambda_r$ be the eigenvalues of $\rho^\Gamma$.} The following lemma, restating Lemma 1 of~\cite{neven2021symmetry}, formalizes this criterion:  
\begin{lemma}\label{lem:ED}  
A quantum state $\rho$ is entangled if
\begin{equation}\label{eq:ptptpt}
e_i(\lambda_1, \ldots, \lambda_r) < 0
\end{equation}
for some $i = 1, 2, \dots, r$, where $p_i^\textnormal{PT}$ are the PT moments of $\rho^\Gamma$, and $e_i(x_1, \dots, x_m)$ denotes the elementary symmetric polynomial in $m$ variables, defined as
\begin{equation}
e_i(x_1, \dots, x_m) = \sum_{1 \leq j_1 < j_2 < \dots < j_i \leq m} x_{j_1} x_{j_2} \cdots x_{j_i},
\end{equation}
which satisfies the recursive formula
\begin{equation}\label{recur:PT}
e_k = \frac{1}{k} \sum_{i=1}^k (-1)^{i-1} e_{k-i} p_i^\textnormal{PT}.
\end{equation}
\end{lemma}
To integrate our approach, suppose that the eigenvalues of $\rho^\Gamma$ are all non-negative. In this case, $\rho^\Gamma$ is a valid density matrix, allowing us to apply \textbf{[Algorithm 2]}. Computing $p_1^\text{PT}, p_2^\text{PT}, \dots, p_t^\text{PT}$ is sufficient to estimate higher-order PT moments, where $t = \mathcal{O}\left(\ln \left({r}/{\epsilon}\right)\right)$. Using these PT moments {and the recursive formula~\cref{recur:PT}}, we compute $e_i(\lambda_1, \ldots, \lambda_r)$ for $i = 1, 2, \dots, r$. If the inequality~\cref{eq:ptptpt} holds for some $i$, then $\rho$ is entangled. Combining~\cref{lem:ED} with our method establishes a new entanglement detection criterion that requires only $p_1^\text{PT}, p_2^\text{PT}, \dots, p_t^\text{PT}$.

We hypothesize that in practical scenarios, the required number of PT moments is significantly smaller. Numerical simulations in~\cref{chap:simul} suggest that $t = 8$ is sufficient in most cases. If experimental validation confirms that $t = 8$ is also adequate for entanglement detection, this could constitute a groundbreaking discovery. We leave quantitative analysis and experimental verification as future work.

\section{Concluding remarks}\label{chap:con}
\subsection{Summary of findings}\label{chap:con-summary}
% rank-dependence
In this paper, we present an efficient algorithm for estimating the trace of quantum state powers. Our first key observation is the discovery of the rank dependence in this problem. Specifically, we find that for a large integer $ k $, estimating $ \T(\rho^k) $ within an additive error $ \epsilon $ requires only the computation of $ \{\T(\rho^i)\}_{i=1}^r $ using quantum resources. The remaining values can then be efficiently estimated within the same additive error $ \epsilon $ by employing a simple recurrence relation based on the Newton-Girard method.

% effective-rank
Our second key observation reveals a condition even stronger than rank dependence. In practical experimental settings, estimating the rank of a given quantum state is often non-trivial and can introduce additional overhead. To address this issue, we introduce the concept of the effective rank $ \widetilde{r} $ and rigorously prove that, for a target power $ k $, our approach requires only quantum resources proportional to $ \ln k $. This result significantly reduces the resource requirements and enhances the feasibility of trace estimation in realistic quantum experiments.

By leveraging the concepts of rank dependence and effective rank, we successfully extended our efficient estimation algorithm to tackle not only the problem of estimating traces of quantum state powers with arbitrary observables, $\T(M\rho^k)$, but also the more general problem of estimating traces of products of quantum state powers, $\T(\rho^k\sigma^l)$. Our main ideas were rigorously validated through formal mathematical proofs and numerical simulations. Furthermore, we demonstrated several practical applications of our algorithm in quantum information processing. Specifically, we showed that it enables more resource-efficient quantum estimation of nonlinear functionals of quantum states, quantum Gibbs state preparation, and entanglement detection compared to previously known methods. Moreover, we illustrated its applicability to various distance measures, further highlighting its broad utility.

\subsection{Future research directions}\label{chap:future}
In our study, several important directions for future investigation remain.
\begin{enumerate}[label=(\arabic*)]
    \item A tighter upper bound on $\abs{\widetilde{P}_k - P_k}$ in~\cref{lem:eff} needs to be established. While~\cref{chap:simul} discusses several observations suggesting the possibility of a tighter bound, a more rigorous mathematical formulation and an explicit analytical expression would be valuable.
    \item A more detailed quantitative analysis of entanglement detection is necessary to identify specific aspects where our algorithm offers concrete improvements.
    \item Perhaps most critically, our current work presents an efficient quantum algorithm for estimating $\T(\rho^k)$ with additive error under multi-copy joint measurements. The algorithm sequentially applies \textsf{CSWAP} gates across registers to construct a fully entangled state over all quantum samples, thereby enabling coherent global operations necessary for accurate estimation. However, the exponential decay of $\T(\rho^k)$ with increasing $k$ suggests that multiplicative-error estimators are necessary in many applications. Despite this importance, our understanding of both additive and multiplicative error dependence remains limited across various measurement models. For the incoherent measurement setting, Liu \textit{et al.}~\cite{liu2024exponential} implicitly establish a lower bound for additive-error estimation, but a corresponding upper bound remains unknown. For multiplicative-error estimation, the first nontrivial upper bound in the fixed-basis incoherent setting was proposed only recently~\cite{pelecanos2025beating}, and no meaningful lower bound is currently available.
    \begin{enumerate}
        \item While the fixed-basis incoherent and coherent measurement settings are now partially understood, relatively little is known about other regimes, including randomized and adaptive measurements, with tight bounds under either additive or multiplicative error still largely missing.
        \item Furthermore, in practically motivated constrained models such as the few-copy measurement setting (where only a small number of copies can be measured jointly) or the bounded quantum memory setting, our understanding is even more limited. While recent work~\cite{chenFOCS24optimal} has characterized the memory-sample tradeoff for Pauli shadow tomography in these restricted settings, extending such analyses to nonlinear estimation tasks remains largely open and would be a particularly exciting direction. Establishing tight bounds in these models continues to pose significant analytical challenges, but offers a rich and important avenue for future work.
    \end{enumerate}
    \item It is necessary to investigate how our proposed algorithm can enhance virtual distillation. In the virtual distillation process, the expectation value of an observable $ M $ with respect to the state $ {\rho^k}/{\operatorname{Tr}(\rho^k)} $ must be computed. Our algorithm is expected to accelerate this computation. However, whether virtual distillation can still yield error-free expectation values under certain conditions when the ideal state $ |\psi\rangle $ cannot be prepared and only a faulty state $ \rho $ is available requires a more detailed analysis.  For our algorithm to be useful in this context, it must be assumed that even if  $|\psi\rangle $ cannot be directly prepared, the quantum computer can still prepare multiple copies of $ \rho $ simultaneously and perform joint operations on them. However, in a faulty quantum computer, the process of preparing multiple copies and performing joint operations may introduce additional errors, which could be larger than those arising in much simpler single-copy operations. Given these considerations, it would be interesting to analyze how our algorithm influences virtual distillation while accounting for these potential error sources.
    \item A natural open problem is to explore how our approach can be extended to general real values of $ k $, rather than just integer $ k $. While the work of~\cite{liu2025estimating} addresses algorithms for non-integer $ k $, it would be interesting to develop an iterative variant in the spirit of our approach. This leads to several open questions, including the challenge of obtaining tighter quantitative bounds.
    \item Beyond the specific problem settings addressed in this paper, it would be highly interesting to explore rank-dependent quantum algorithms applicable to broader areas of quantum information processing, particularly those that are especially efficient for low-rank quantum states.
\end{enumerate}

\section*{Data availability statement}
The data and software that support the findings of
this study can be found in the following repository:
\url{https://github.com/tfoseel/trace-of-powers}

\section*{Acknowledgments}
J.L. thanks Chirag Wadhwa for helpful discussions and feedback regarding the importance of considering the trace of powers problem under various measurement settings. The authors also thank the anonymous reviewers for their valuable comments that helped improve the manuscript. This work was supported by the National Research Foundation of Korea (NRF) through a grant funded by the Ministry of Science and ICT (Grant Nos. RS-2023-00211817; RS-2024-00404854; RS-2025-00515537). This work was also supported by the Institute for Information \& Communications Technology Promotion (IITP) grants funded by the Korean government (MSIP) (Grant Nos. RS-2025-02304540; 2019-II190003, Research and Development of Core Technologies for Programming, Running, Implementing, and Validating of Fault-Tolerant Quantum Computing Systems), the National Research Council of Science \& Technology (NST) (Grant No. GTL25011-401), and Korea Institute of Science and Technology Information (KISTI: P25026).

\section*{Author Contributions}
M.S. and J.L. contributed equally to this work, undertaking the primary responsibilities, including the development of the main ideas, mathematical proofs, initial drafting, and revisions of the paper. S.L. contributed to the numerical simulations and paper preparation. K.J. supervised the research. All authors discussed the results and contributed to the final paper.

\bibliographystyle{quantum}
\bibliography{references}

\newpage
\appendix
\section{Omitted proofs}\label{chap:proofs}
\subsection{Proof of~\cref{lem:rank}}\label{lem:rank:proof}
\begin{proof}
Let's try to find some properties of $\abs{d_i}$.
\begin{itemize}
    \item If $i=1$,
\begin{align}
    d_1 &= b_1 - a_1 = Q_1 - P_1 = \epsilon_1,
\end{align}
this gives $\abs{d_1}\le\abs{\epsilon_1}$.
    \item And if $i=2$,
\begin{align}
    d_2 &= b_2 - a_2 \nonumber \\
    &=\frac{Q_1^2 - Q_2}{2} - \frac{P_1^2 - P_2}{2} = \frac{(Q_1 + P_1)(Q_1 - P_1) - (Q_2 - P_2)}{2} = \epsilon_1 - \frac{1}{2} \epsilon_2,
\end{align}
this gives $\abs{d_2}\le\abs{\epsilon_1} + \frac{\abs{\epsilon_2}}{2}$.
    \item And if $i=3$,
\begin{align}
    d_3 &= \frac{(b_2 - a_2)Q_1 - (b_1 - a_1)Q_2 + (Q_3 - P_3)-a_1(Q_2 - P_2) + a_2(Q_1 - P_1)}{3} .
\end{align}
this gives,
\begin{align}
    \abs{d_3} &\le \frac{\abs{d_2} P_1 + \abs{d_1} P_2 + \abs{\epsilon_3} + a_1 \abs{\epsilon_2} + a_2 \abs{\epsilon_1}}{3} \le \abs{\epsilon_1} + \frac{\abs{\epsilon_2}}{2} + \frac{\abs{\epsilon_3}}{3}.
\end{align}
\end{itemize}

Now, we suppose $\abs{d_k} \le \sum_{i=1}^k \frac{\abs{\epsilon_i}}{i}$. Then,
\begin{align}
    d_{k+1} &= {\frac{1}{k+1}} \Bigg\{\sum_{i=1}^{k+1} (-1)^{i-1} b_{k+1-i} Q_i - \sum_{i=1}^{k+1} (-1)^{i-1} a_{k+1-i} P_i\Bigg\} \nonumber\\
    &= {\frac{1}{k+1}} \Bigg\{\sum_{i=1}^{k+1} (-1)^{i-1} (b_{k+1-i} - a_{k+1-i}) Q_i + \sum_{i=1}^{k+1} (-1)^{i-1} a_{k+1-i} (Q_i - P_i)\Bigg\}.
\end{align}
Taking the absolute value, and by using $\abs{Q_i}\le 1$ and $\abs{a_{k+1-i}}\le 1$,
\begin{align}
    \abs{d_{k+1}} &\le {\frac{1}{k+1}} \Bigg(\sum_{i=1}^k \abs{b_{k+1-i} - a_{k+1-i}} + \sum_{i=1}^{k+1} \abs{Q_i - P_i}\Bigg) \nonumber\\
    & \le {\frac{1}{k+1}} \left(\sum_{i=1}^k \sum_{j=1}^{k+1-i} \frac{\abs{\epsilon_j}}{j} + \sum_{i=1}^{k+1} \abs{\epsilon_i}\right) \nonumber\\
    & = {\frac{1}{k+1}} \sum_{j=1}^{k+1} \left(\frac{k+1-j}{j} + 1 \right) \abs{\epsilon_j} \nonumber\\
    & = \sum_{j=1}^{k+1} \frac{\abs{\epsilon_j}}{j}.
\end{align}
So, we can conclude that
\begin{equation}
    \forall k\in \mathbb{N},~\abs{d_k}\le \sum_{i=1}^k \frac{\abs{\epsilon_i}}{i}    
\end{equation}
by strong mathematical induction logic.
\end{proof}

\subsection{Proof of~\cref{thm:rank}}\label{thm:rank:proof}
\begin{proof}
We assume
\begin{equation}
    \abs{\epsilon_i} < \frac{\epsilon}{kt\ln t}    
\end{equation}
holds for $i=1,2,...,t$. We first prove that
\begin{equation}
    \abs{Q_i-\widetilde{P}_i} < \epsilon    
\end{equation}
for $i=1,2,...,k$ always holds.
Consider the recurrence relations defined by the~\cref{DefQ}, and~\cref{p'recurr}. Then the difference between $\widetilde{P}_{t+k}, Q_{t+k}$ becomes:
\begin{align}
    Q_{t+k}-\widetilde{P}_{t+k} =\sum_{j=1}^t(-1)^{j-1}a_j\left(Q_{t+k-j}-\widetilde{P}_{t+k-j}\right) +\sum_{j=1}^t(-1)^{j-1}(b_j-a_j)Q_{t+k-j}.
\end{align}
Let $\widetilde{\epsilon}_i = Q_i-\widetilde{P}_i$ for all $i$. Then,
\begin{align}
    \widetilde{\epsilon}_{t+k}&=\sum_{j=1}^t \Bigl\{(-1)^{j-1} \left(a_j \widetilde{\epsilon}_{t+k-j}+d_jQ_{t+k-j}\right)\Bigl\},\\
    \widetilde{\epsilon}_{t+k-1}&=\sum_{j=1}^t \Bigl\{(-1)^{j-1} \left(a_j \widetilde{\epsilon}_{t+k-j-1} + d_jQ_{t+k-j-1}\right)\Bigl\}.
\end{align}
By exploiting $a_1 = 1$, we can sum up the above expressions in the form of
\begin{align}
    \widetilde{\epsilon}_{t+k} &=\sum_{j=2}^t (-1)^{j-1} a_j \widetilde{\epsilon}_{t+k-j} + \sum_{j=1}^t (-1)^{j-1} a_j \widetilde{\epsilon}_{t+k-j-1} + \sum_{j=1}^t (-1)^{j-1} d_j (Q_{t+k-j}+Q_{t+k-j-1}) \nonumber\\
    & = \sum_{j=1}^{t-1} (-1)^{j-1} (a_{j}-a_{j+1})\widetilde{\epsilon}_{r+k-j-1} + (-1)^{t-1} a_t \widetilde{\epsilon}_k + \sum_{j=1}^t (-1)^{j-1} d_j (Q_{t+k-j}+Q_{t+k-j-1}).
\end{align}
Since $a_j \ge a_{j+1}~(\text{trivial from~\cref{chap:NG}})$,
\begin{align}
    \abs{\widetilde{\epsilon}_{t+k}} & \le \sum_{j=1}^{t-1} (a_j-a_{j+1}) \abs{\widetilde{\epsilon}_{t+k-j-1}} + a_t \abs{\widetilde{\epsilon}_k} + \sum_{j=1}^t \abs{d_j} (Q_{t+k-j}+Q_{t+k-j-1})\\
    & \le  \sum_{j=1}^{t-1} (a_j-a_{j+1}) \abs{\widetilde{\epsilon}_{t+k-j-1}} + a_t \abs{\widetilde{\epsilon}_k} + 2\sum_{j=1}^t \abs{d_j}.
\end{align}
Note that, from~\cref{p'}, $P_i = \widetilde{P}_i$ for $i=1,2,...,t$. Therefore, $\epsilon_i = \tilde{\epsilon}_i$ for $i=1,2,...,t$. Let $\epsilon' := \max_{1\le j\le t} \abs{\epsilon_j}$. Suppose that

\begin{equation}\label{eq:induction}
    \widetilde{\epsilon}_{t+m} \le \epsilon' + m\sum_{j=1}^t \abs{d_j}
\end{equation}
holds for $m=1, 2, \ldots, k-1$. Then,
\begin{align}
    \abs{\widetilde{\epsilon}_{t+k}} &\le \sum_{j=1}^{t-1} \left(a_j - a_{j+1}\right) \Biggl\{\epsilon' + (k-2)\sum_{j=1}^t \abs{d_j} \Biggl\} + a_t \Biggl\{\epsilon' + (k-2)\sum_{j=1}^t \abs{d_j}\Biggl\} + 2\sum_{j=1}^t \abs{d_j} \nonumber\\
    & \le a_1 \Biggl\{\epsilon' + (k-2)\sum_{j=1}^t \abs{d_j} \Biggl\} + 2\sum_{j=1}^t \abs{d_j}\nonumber\\
    &= \epsilon' + k\sum_{j=1}^t \abs{d_j}.
\end{align}
Moreover, $m = k$ also holds. Since $m = 0$ trivially holds, by strong mathematical induction logic, for every $m$,~\cref{eq:induction} holds. By applying~\cref{lem:rank}, we get: 
\begin{equation}
    \abs{d_j} \leq \sum_{i=1}^j \frac{\epsilon'}{i} \leq \epsilon' \ln j.
\end{equation}
Finally,
\begin{equation}
    \abs{\widetilde{\epsilon}_{t+k}} \leq \epsilon'+\epsilon'k\ln t! < \epsilon'kt\ln t < \epsilon.
\end{equation}
We proved $\abs{Q_i-\widetilde{P}_i}<\epsilon$ for $i=1,2,...,k$. To conclude the proof, we set $t=r$. We will show that if $t=r$, $P_i=\widetilde{P}_i$ holds for all $i$. Consider the following polynomial,
\begin{align}
    x^r-a_1x^{r-1}+a_2x^{r-2}-\ldots+(-1)^ra_r =(x-p_1)(x-p_2)\ldots(x-p_r).
\end{align}
Then,
\begin{align}
    p_i^{r+k} &= \sum_{j=1}^r (-1)^{j-1} a_j p_i^{r+k-j},
\end{align}
And we have,
\begin{align}
    P_{r+k} = \sum_{j=1}^r (-1)^{j-1} a_j P_{r+k-j},
\end{align}
which is the same recurrence relation with~\cref{p'recurr}, when $t=r$. Hence, $P_i=\widetilde{P}_i$ and
\begin{equation}
    \abs{\epsilon_i}=\abs{Q_i-P_i}=\abs{Q_i-\widetilde{P}_i}<\epsilon
\end{equation}
for $i=1,2,\dots,r$ which completes the proof.
\end{proof}

\subsection{Proof of~\cref{cor:rank}}\label{cor:rank:proof}
\begin{proof}
Using the multivariate trace estimation method~\cite{quek2024multivariate}, it is known that with 
\begin{equation}
    \mathcal{O}\left({\frac{\ln(1/\delta)}{\epsilon^{2}}}\right)
\end{equation}
runs on a constant-depth quantum circuit consisting of $\mathcal{O}(i)$ qubits and $\mathcal{O}(i)$ \textsf{CSWAP} operations, we can estimate each $\T(\rho^i)$ within an $\epsilon$ additive error and with a success probability of no less than $1-\delta$. Note that only the maximum error
\begin{equation}
    \epsilon' := \max_{1\le j\le t} \abs{\epsilon_j}    
\end{equation}
affects the error of our algorithm. Thus, with
\begin{equation}
    \mathcal{O}\left({\frac{k^2t^2\ln^2 t\ln(1/\delta)}{\epsilon^{2}} }\right)
\end{equation}
runs, we can satisfy the assumption in~\cref{thm:rank}. Hence, $\T(\rho^i)$ $(\forall i \leq k)$ can be estimated within an $\epsilon$ error and with a success probability of no less than $1-\delta$. Finally, we set $t=r$ in~\cref{thm:rank}, which concludes the proof.  
\end{proof}

\subsection{Proof of~\cref{lem:eff}}\label{lem:eff:proof}
\begin{proof}
Let $\epsilon_i = \widetilde{P}_i - P_i$. By definition, we have $\epsilon_i = 0$ for $i = 1, 2, \dots, t$. For $i \geq t+1$, we have
\begin{align}
    \epsilon_{t+k} = \widetilde{P}_{t+k} - P_{t+k}.    
\end{align}
From the definition of $\widetilde{P}_{t+k}$, we have:
\begin{align}
    \epsilon_{t+k} = \left\{\sum_{i=1}^t (-1)^{i-1} \widetilde{P}_{t+k-i}a_i\right\} - P_{t+k}.    
\end{align}
Rewriting this, we split $\epsilon_{t+k}$ into two terms:
\begin{align}
    \epsilon_{t+k}  = \left\{\sum_{i=1}^t (-1)^{i-1} \epsilon_{t+k-i}a_i\right\}  + \left\{\sum_{i=1}^t (-1)^{i-1} P_{t+k-i}a_i - P_{t+k}\right\}.    
\end{align}
Define the second term as $z_{t+k}$ for brevity:
\begin{align}
    z_{t+k} &= \left\{\sum_{i=1}^t (-1)^{i-1} P_{t+k-i}a_i - P_{t+k}\right\} \\
    & = (-1)^{t-1} \sum_{\{\alpha_1, \ldots, \alpha_{t+1}\} \subseteq [r]} \left( \sum_{i=1}^{t+1} p_{\alpha_i}^k \right) \prod_{i=1}^{t+1} p_{\alpha_i}.
\end{align}
Thus,
\begin{align}
    \epsilon_{t+k} = \sum_{i=1}^t (-1)^{i-1} \epsilon_{t+k-i}a_i + z_{t+k}.    
\end{align}
We first bound $z_{t+k}$:
\begin{align}
    \abs{z_{t+k}} & \leq (t+1)\ \sum_{\{\alpha_1, \ldots, \alpha_{t+1}\} \subseteq [r]} \prod_{i=1}^{t+1} p_{\alpha_i} = (t+1)a_{t+1}.    
    \end{align}
    Next, consider the recursive relation for $\epsilon_{t+k}$:
    \begin{align}
    \epsilon_{t+k} = \sum_{i=1}^t (-1)^{i-1} \epsilon_{t+k-i}a_i + z_{t+k}.    
    \end{align}
    Combining this with the relation for $\epsilon_{t+k-1}$, we get:
    \begin{align}
    \epsilon_{t+k} & = \left\{\sum_{i=1}^{t-1} (-1)^{i-1} \epsilon_{t+k-i-1}(a_i - a_{i+1})\right\} + \epsilon_{k-1}a_t + z_{t+k} + z_{t+k-1}.
    \end{align}
Taking the absolute value, we bound $\abs{\epsilon_{t+k}}$:
\begin{align}
    \abs{\epsilon_{t+k}} & \leq \left\{\sum_{i=1}^{t-1} \abs{\epsilon_{t+k-1-i}}(a_i - a_{i+1})\right\} + \abs{\epsilon_{k-1}}a_t + 2(t+1)a_{t+1} \nonumber \\
    & \le \epsilon_{\text{max}} + 2(t+1)a_{t+1},
\end{align}
where $\epsilon_{\text{max}}$ defined as:
\begin{align}
    \epsilon_{\text{max}} = \max_{i \leq t+k-2} \abs{\epsilon_i}.
\end{align}
By induction, we conclude:
\begin{align}
    \abs{\epsilon_{t+k}} \leq k(t+1)a_{t+1}.    
\end{align}
Also, we can bound $a_{t+1}$ as follows (see the details in~\cref{chap:detail-a}):
\begin{align}
    a_{t+1} & \leq \binom{r}{t+1} \frac{1}{r^{t+1}}  = \frac{r(r-1)\ldots(r-t-1)}{(t+1)! r^{t+1}} \leq \frac{1}{(t+1)!}\left(1-\frac{t}{r}\right).
\end{align}
Combining the results, we have:
\begin{align}
    \abs{\epsilon_k} = \abs{P_k - \widetilde{P}_k}\leq (k-t)(t+1)a_{t+1}\leq k(t+1)a_{t+1}.
\end{align}
Substituting the bound for $a_{t+1}$, we get:
\begin{align}
    \abs{\epsilon_k} \leq \frac{k}{t!}\left(1-\frac{t}{r}\right).
\end{align}
\end{proof}

\subsubsection{Bounding $a_{t}$}\label{chap:detail-a}
Define $ A_{j,k} $ as:  
\begin{equation}
    A_{j,k} = \sum_{\{\alpha_1, \ldots, \alpha_t\} \subseteq [k]} \left( \prod_{i=1}^{j} p_{\alpha_i} \right),
\end{equation}  
where $ p_i \geq 0 $ and $ \sum_{i=1}^k p_i = 1 $. By definition, we have $ a_t = A_{t,r} $, where $ r $ denotes the rank. Next, let $ x = p_1 $ and define  
\[
    p_i' = \frac{p_{i+1}}{1-x}, \quad \text{for } i = 1,2,\dots,k-1.
\]  
Note that $ \sum_{i=1}^{k-1} p_i' = 1 $. Define $ A_{j,k}' $ as:  
\begin{equation}
    A_{j,k}' = \sum_{\{\alpha_1, \ldots, \alpha_t\} \subseteq [k]} \left( \prod_{i=1}^{j} p_{\alpha_i}' \right).
\end{equation}  
Importantly, $ x $ and $ A_{j,k}' $ are independent. For $ (j,k) < (t,r) $, suppose that $ A_{j,k} $ is maximized when  
\begin{equation}
    p_1 = p_2 = \dots = p_k = \frac{1}{k}.    
\end{equation}
We then obtain the recurrence relation:
\begin{equation}
    A_{t,r} = x(1-x)^{t-1} A_{t-1,r-1}' + (1-x)^t A_{t,r-1}'.
\end{equation}  
Since it is straightforward to verify that  
\begin{equation}
    \max_{p_i'} A_{j,k}' = \max_{p_i} A_{j,k},    
\end{equation}
it follows that (by assumption) $ A_{t-1,r-1}' $ and $ A_{t,r-1}' $ are maximized when  
\begin{equation}
    p_1' = p_2' = \dots = p_{r-1}' = \frac{1}{r-1}.    
\end{equation}
Thus, we obtain  
\begin{align}
    \max A_{t-1,r-1}' &= \binom{r-1}{t-1} \frac{1}{(r-1)^{t-1}},\\
    \max A_{t,r-1}' &= \binom{r-1}{t} \frac{1}{(r-1)^t}.    
\end{align}
Now, considering the maximization over $ x $ and $ p_i' $, we derive:
\begin{align}
    \max_{p_i} A_{t,r} 
    &= \max_{x, p_i'} A_{t,r} \nonumber \\ \nonumber
    &= \max_x \bigg\{ x(1-x)^{t-1} \max_{p_i'} A_{t-1,r-1}' + (1-x)^t \max_{p_i'} A_{t,r-1}' \bigg\} \\ \nonumber
    &= \max_x \bigg\{ x(1-x)^{t-1} \binom{r-1}{t-1} \frac{1}{(r-1)^{t-1}} + (1-x)^t \binom{r-1}{t} \frac{1}{(r-1)^t} \bigg\} \\ \nonumber
    &= \frac{1}{r(r-1)^t} \binom{r}{t} \times \max_x \bigg\{ (1-x)^{t-1} (r - rx - t + rtx) \bigg\}.
\end{align}  
Define  
\begin{equation}
    f(x) = (1-x)^{t-1} (r - rx - t + rtx).    
\end{equation}
Differentiating $ f(x) $ and solving for its maximum, we find that the optimal value occurs at $ x = {1}/{r} $. This implies that $ A_{t,r} $ is maximized when
\begin{align}
    & p_1 = x = \frac{1}{r}, \\
    & p_2 = p_3 = \dots = p_r = \frac{1}{r-1} (1-x) = \frac{1}{r}.
\end{align}
By strong induction, we conclude that for all $ t, r $,
\begin{equation}
a_t \leq \max A_{t,r} = \binom{r}{t} \frac{1}{r^t}.
\end{equation}

\subsection{Proof of~\cref{thm:eff}}\label{thm:eff:proof}
\begin{proof}
We adopt the same notation as in~\cref{lem:eff:proof}. Using the proof from~\cref{thm:rank:proof}, we conclude that if  
\begin{align}
    \epsilon' = \frac{\epsilon}{2kt\ln t},    
\end{align}  
then the following condition holds:
\begin{align}
    \abs{\widetilde{P}_i - Q_i} < \frac{\epsilon}{2}.    
\end{align}
To estimate $P_i$ using $Q_i$ with an additive error $\epsilon$, we must ensure that:
\begin{align}
    \abs{\widetilde{P}_i - P_i} < \frac{\epsilon}{2}.    
\end{align}
This ensures:
\begin{align}
    \abs{P_i - Q_i} \leq \abs{\widetilde{P}_i - Q_i} + \abs{\widetilde{P}_i - P_i} < \epsilon,    
\end{align}
for all $i = 1, 2, \dots, k$. From~\cref{lem:eff}, we derive the following bound:
\begin{align}\label{proof:approx}
    \abs{\widetilde{P}_i - P_i}=\abs{\epsilon_i} \leq \frac{i}{t!}\left(1-\frac{t}{r}\right) \leq \frac{k}{e^t}.    
\end{align}
To ensure $\abs{\widetilde{P}_i - P_i} < {\epsilon}/{2}$, it suffices to satisfy:
\begin{align}
    \frac{k}{e^t} < \frac{\epsilon}{2}.    
\end{align}
Taking logarithms and rearranging terms, we obtain:
\begin{align}
    t > \ln\left(\frac{2k}{\epsilon}\right) \ge \left\lceil \ln\left(\frac{2k}{\epsilon}\right) \right\rceil.    
\end{align}
\end{proof}

\subsection{Proof of~\cref{thm:obs}}\label{thm:obs:proof}
\begin{proof}
Let
\begin{equation}
    \rho = \sum_{i=1}^r p_i \ket{\psi_i} \bra{\psi_i},
\end{equation}
and
\begin{equation}
    m_i = \bra{\psi_i} M \ket{\psi_i}.
\end{equation}
We introduce a new quantity, denoted as $\widetilde{P}_{i,M}$, defined for $i \leq t$ as follows:
\begin{equation}\label{p'-obs}
    \widetilde{P}_{i(\leq t),M} := \text{Tr}(M\rho^{i}) = \sum_{j=1}^r m_jp_{j}^i.
\end{equation}  
For $i > t$, $\widetilde{P}_{i,M}$ is recursively defined based on the Newton-Girard recurrence relations, where the elementary symmetric polynomials $a_k$ are defined in~\cref{elepoly1}.
\begin{equation}\label{p'recurr-obs}  
    \widetilde{P}_{i(> t),M} := \sum_{k=1}^t (-1)^{k-1} a_k \widetilde{P}_{i-k,M}.  
\end{equation}  
We assume that
\begin{equation}
    \abs{\epsilon_{i,M}} < \frac{\epsilon}{4},    
\end{equation}
and
\begin{equation}
    \abs{\epsilon_i} < \frac{\epsilon}{4\norm{M}_{\infty}kt\ln t},
\end{equation}
for $i = 1, 2, \dots, t$. We will first prove that
\begin{equation}
    \abs{Q_{i,M}-\widetilde{P}_{i,M}} < \frac{\epsilon}{2}    
\end{equation}
holds for $i = 1, 2, \dots, k$. The difference between $\widetilde{P}_{i,M}$ and $Q_{i,M}$ is given by:
\begin{align}
    Q_{t+k,M}-\widetilde{P}_{t+k,M} =\sum_{i=1}^{t}(-1)^{j-1}a_j\left(Q_{t+k-j,M}-\widetilde{P}_{t+k-j,M}\right) +\sum_{i=1}^t(-1)^{j-1}(b_j-a_j)Q_{t+k-j,M}.
\end{align}
Let $\widetilde{\epsilon}_{i,M} := Q_{i,M} - \widetilde{P}_{i,M}$, so that we can write:
\begin{align}
    \widetilde{\epsilon}_{t+k,M} &=\sum_{j=1}^t \left\{(-1)^{j-1} \left(a_j \widetilde{\epsilon}_{t+k-j,M}+d_jQ_{t+k-j,M}\right)\right\},\\
    \widetilde{\epsilon}_{t+k-1,M} &=\sum_{j=1}^t \left\{(-1)^{j-1} \left(a_j \widetilde{\epsilon}_{t+k-j-1,M} + d_jQ_{t+k-j-1,M}\right)\right\}.
\end{align}
We define
\begin{align}
    \epsilon' &= \max_{1\le j\le t} \abs{\epsilon_{j}},\quad\epsilon'_M = \max_{1\le j\le t} \abs{\epsilon_{j,M}}.
\end{align}
and by assumption, we have
\begin{equation}
    \epsilon'<\frac{\epsilon}{4\norm{M}_{\infty}kt\ln t}    
\end{equation}
and $\epsilon'_M < {\epsilon}/{4}$. Using the same logic as in the proof of~\cref{thm:rank}, and noting that $Q_{i,M} \leq \norm{M}_{\infty}$, we conclude that for every $k$
\begin{equation}\label{induction}
    \abs{\widetilde{\epsilon}_{t+k,M}} \leq \epsilon'_M + k\norm{M}_{\infty}\sum_{j=1}^t\abs{d_j}
\end{equation}
holds. By applying~\cref{lem:rank}, we can conclude that
\begin{equation}
    \abs{\widetilde{\epsilon}_{t+k,M}} \leq \epsilon'_M + \epsilon'kt\ln t\norm{M}_{\infty} < \frac{\epsilon}{2}.
\end{equation}
Therefore, $\abs{Q_{i,M}-\widetilde{P}_{i,M}}<{\epsilon}/{2}$ holds for $i=1,2,...,k$. \\~\\
Next, we aim to prove that $\abs{\widetilde{P}_{i,M}-P_{i,M}}<{\epsilon}/{2}$. Note that
\begin{equation}
    t = \ln\left(\frac{2k\norm{M}_{\infty}}{\epsilon}\right).
\end{equation}
Let $\delta_i = \widetilde{P}_{i,M} - P_{i,M}$. By definition, $\delta_i = 0$ for $i = 1, 2, \dots, t$. For $i \geq t+1$, we derive $\delta_{t+k}$ as follows:
\begin{align}
    \delta_{t+k} = \widetilde{P}_{t+k,M} - P_{t+k,M}.    
\end{align}
Using the definition of $\widetilde{P}_{t+k,M}$, we get
\begin{align}
    \delta_{t+k} = \left\{\sum_{i=1}^t (-1)^{i-1} \widetilde{P}_{t+k-i,M}a_i\right\} - P_{t+k,M}.    
\end{align}
Rewriting this expression, we split $\delta_{t+k}$ into two terms:
\begin{align}
    \delta_{t+k} & = \left\{\sum_{i=1}^t (-1)^{i-1} \delta_{t+k-i}a_i\right\}  + \left\{\sum_{i=1}^t (-1)^{i-1} P_{t+k-i,M}a_i - P_{t+k,M}\right\}.    
\end{align}
We define the second term as $z_{t+k}$ for brevity:
\begin{align}
    z_{t+k} &= \left\{\sum_{i=1}^t (-1)^{i-1} P_{t+k-i,M}a_i - P_{t+k,M}\right\} \\
    & = (-1)^{t-1} \sum_{\{\alpha_1, \ldots, \alpha_{t+1}\} \subseteq [r]} \left( \sum_{i=1}^{t+1} m_{\alpha_i}p_{\alpha_i}^k \right) \prod_{i=1}^{t+1} p_{\alpha_i}.
\end{align}
Thus, we have
\begin{align}
    \delta_{t+k} = \sum_{i=1}^t (-1)^{i-1} \delta_{t+k-i}a_i + z_{t+k}.    
\end{align}
We first bound $z_{t+k}$:
\begin{align}
    \abs{z_{t+k}} & \leq (t+1)\norm{M}_{\infty} \sum_{\{\alpha_1, \ldots, \alpha_{t+1}\} \subseteq [r]} \prod_{i=1}^{t+1} p_{\alpha_i} \nonumber \\
    & = (t+1)\norm{M}_{\infty}a_{t+1}.    
\end{align}
Using the same induction logic as in the proof of~\cref{lem:eff}, we get:
\begin{equation}
    \abs{\delta_{i}} \leq i(t+1)\norm{M}_{\infty}a_{t+1}.
\end{equation}
Since 
\begin{align}
    a_{t+1} \leq \frac{1}{(t+1)!},\quad t! \geq e^t, \quad t =\ln\left(\frac{2k\norm{M}_{\infty}}{\epsilon}\right),
\end{align}
we conclude that
\begin{equation}
    \abs{P_{i,M}-\widetilde{P}_{i,M}} = \abs{\delta_i} \leq \frac{i\norm{M}_{\infty}}{2^t} <\frac{\epsilon}{2}
\end{equation}
holds for $i=1,2,...,k$.\\~\\
Thus, $\abs{Q_{i,M} - P_{i,M}} < \epsilon$ holds for $i = 1, 2, \dots, k$.
\end{proof}

\subsection{Proof of~\cref{cor:obs}}\label{cor:obs:proof}
\begin{proof}
Using the multivariate trace estimation method~\cite{quek2024multivariate}, it is known that with 
\begin{equation}
    \mathcal{O}\left({\frac{\ln(1/\delta)}{\epsilon^{2}}}\right)
\end{equation}
runs on a constant-depth quantum circuit consisting of $\mathcal{O}(i)$ qubits and $\mathcal{O}(i)$ \textsf{CSWAP} operations, we can estimate each $\T(\rho^i)$ within $\epsilon$ additive error and with success probability not smaller than $1-\delta$.\\
Note that only the maximum error
\begin{equation}
    \epsilon' := \max_{1 \le j \le t} \abs{\epsilon_j},   
\end{equation}
and
\begin{equation}
    \epsilon_M' := \max_{1 \le j \le t} \abs{\epsilon_{j,M}},   
\end{equation}
affects the error of our algorithm. So with
\begin{equation}
    \mathcal{O}\left({\frac{k^2\norm{M}_\infty t^2\ln^2 t\ln(1/\delta)}{\epsilon^{2}}}\right)
\end{equation}
runs for estimating $\T(\rho^{j'})$ $(j' \leq r)$, and 
\begin{equation}
    \mathcal{O}\left({\frac{c^2N_{M}\ln(1/\delta)}{\epsilon^2}}\right)    
\end{equation}
runs for estimating $\T(M\rho^j)$ $(j \leq r)$, we can satisfy the assumption in~\cref{thm:obs} with success probability not smaller than $1-\delta$. Hence, $\T(M\rho^i)$ $(\forall i \leq k)$ can be estimated within $\epsilon$ error and with success probability not smaller than $1-\delta$. Finally, we set $t=\widetilde{r}_M$ in~\cref{thm:obs}, which concludes the proof by ignoring the logarithmic terms.
\end{proof}

\subsection{Proof of~\cref{thm:multi}}\label{thm:multi:proof}
\begin{proof}
For a fixed index $i \leq t$, we begin by estimating the sets $ \{\T(\sigma^j)\}_{j=1}^t $ and $ \{\T(\rho^i\sigma^j)\}_{j=1}^{t} $, which enables the application of \textbf{{[Algorithm 2]}} with $M = \rho^i$. By~\cref{thm:obs}, we obtain:
\begin{equation}\label{a8-1}
    \{\T(\sigma^j)\}_{j=1}^t~~\text{within additive error}~~\frac{\epsilon}{8lt\ln t},
\end{equation}
and 
\begin{equation}\label{a8-2}
    \{\T(\rho^i\sigma^j)\}_{j=1}^{t}~~\text{within additive error}~~\frac{\epsilon}{8}.
\end{equation}
We then set
\begin{equation}\label{a8-3}
    t \geq \widetilde{R} = \min\left\{r, \left\lceil\ln\left(\frac{4k+4l}{\epsilon}\right)\right\rceil\right\}    
\end{equation}
which allows us to compute the estimated values for
\begin{equation}\label{a8-4}
    \{\T(\rho^i\sigma^j)\}_{j=1}^l~~\text{within additive error}~~\frac{\epsilon}{2},
\end{equation}
for $i=1,2,...t$. For a fixed index $j$, obtaining 
\begin{equation}\label{a8-5}
    \{\T(\rho^i)\}_{i=1}^t~~\text{within additive error}~~\frac{\epsilon}{4kt\ln t},
\end{equation}
and using~\cref{a8-4} enables the application of {\textbf{[Algorithm 2]}} with $M = \sigma^j$. This, in turn, allows the computation of the estimated values for
\begin{equation}
    \{\T(\rho^i\sigma^j)\}_{i=1}^k~~\text{within additive error}~~\epsilon,
\end{equation} 
for $j=1,2,...,l$. To satisfy~\cref{a8-1}, we need $\widetilde{\mathcal{O}}\left({l^2}/{\epsilon^2}\right)$ copies of $\sigma$, and to satisfy~\cref{a8-5}, we need $\widetilde{\mathcal{O}}\left({k^2}/{\epsilon^2}\right)$ copies of $\rho$. The contributions to the copy complexity from other conditions, such as~\cref{a8-2} and~\cref{a8-4}, are analogous. Therefore, to follow the algorithm, we can conclude that the required number of copies of $\rho$ is $\widetilde{\mathcal{O}}\left({k^2}/{\epsilon^2}\right)$ and the required number of copies of $\sigma$ is $\widetilde{\mathcal{O}}\left({l^2}/{\epsilon^2}\right)$. Here, the notation $\widetilde{\mathcal{O}}(\cdot)$ hides polylogarithmic factors in $k$ and $l$. So, setting $t \geq \widetilde{R}$ is sufficient.
\end{proof}

\section{Additional numerical simulations}\label{chap:add-sim}
In~\cref{chap:simul-result}, we anticipated that a lower bound on $t$ could be expressed as  
\begin{equation}  
    \mathcal{O}\left(\frac{\ln\left({k}/{\epsilon}\right)}{\ln \ln \left({k}/{\epsilon}\right)}\right).    
\end{equation}  
While a rigorous mathematical proof remains an open problem for future research, we conducted experiments under Scenario 1, as described in~\cref{chap:simul-setup}, by setting
\begin{equation}\label{eq:new-bound}
    t = \min\left\{r, \left\lceil \frac{\ln(k/\epsilon)}{\ln\ln(k/\epsilon)} \right\rceil\right\}    
\end{equation}
and evaluating four different distributions. (The values of $t$ for different $ (k, \epsilon) $ are listed in~\cref{tab:sim-1-2}.) The results, presented in~\cref{fig:ln-new-sim}, show that although the error is larger compared to when $\widetilde{r}$ was used, the estimation still successfully remains below the target additive error in all cases.

\newpage
\begin{table}[htbp!]
\centering
\resizebox{0.6\columnwidth}{!}{%
\begin{tabular}{c|ccccccc}
$(k,\epsilon)$ & $10^{-1}$ & $10^{-2}$ & $10^{-3}$ & $10^{-4}$ & $10^{-5}$ & $10^{-6}$ & $10^{-7}$ \\ \hline
$8$   & 3         & 4         & 5         & 5         & 6         & 6         & 7         \\
$16$  & 4         & 4         & 5         & 5         & 6         & 6         & 7         \\
$32$  & 4         & 4         & 5         & 5         & 6         & 7         & 7         \\
$64$  & 4         & 5         & 5         & 6         & 6         & 7         & 7         \\
$128$ & 4         & 5         & 5         & 6         & 6         & 7         & 7         \\
$256$ & 4         & 5         & 5         & 6         & 7         & 7         & 8        
\end{tabular}%
}
\caption{{\textbf{The value of $t$ as a function of $(k, \epsilon)$ in~\cref{chap:add-sim}.} The value of $ t $ used in additional numerical simulation is $\min\left\{r, \left\lceil {\ln(k/\epsilon)}/{\ln\ln(k/\epsilon)} \right\rceil\right\} $. Note that $ r = 16 $.}}
\label{tab:sim-1-2}
\end{table}

\begin{figure}[t!]
    \centering
    % 첫 번째 행
    \begin{subfigure}{0.49\linewidth}
        \centering
        \caption{Geometrically decaying}
        \includegraphics[width=\linewidth]{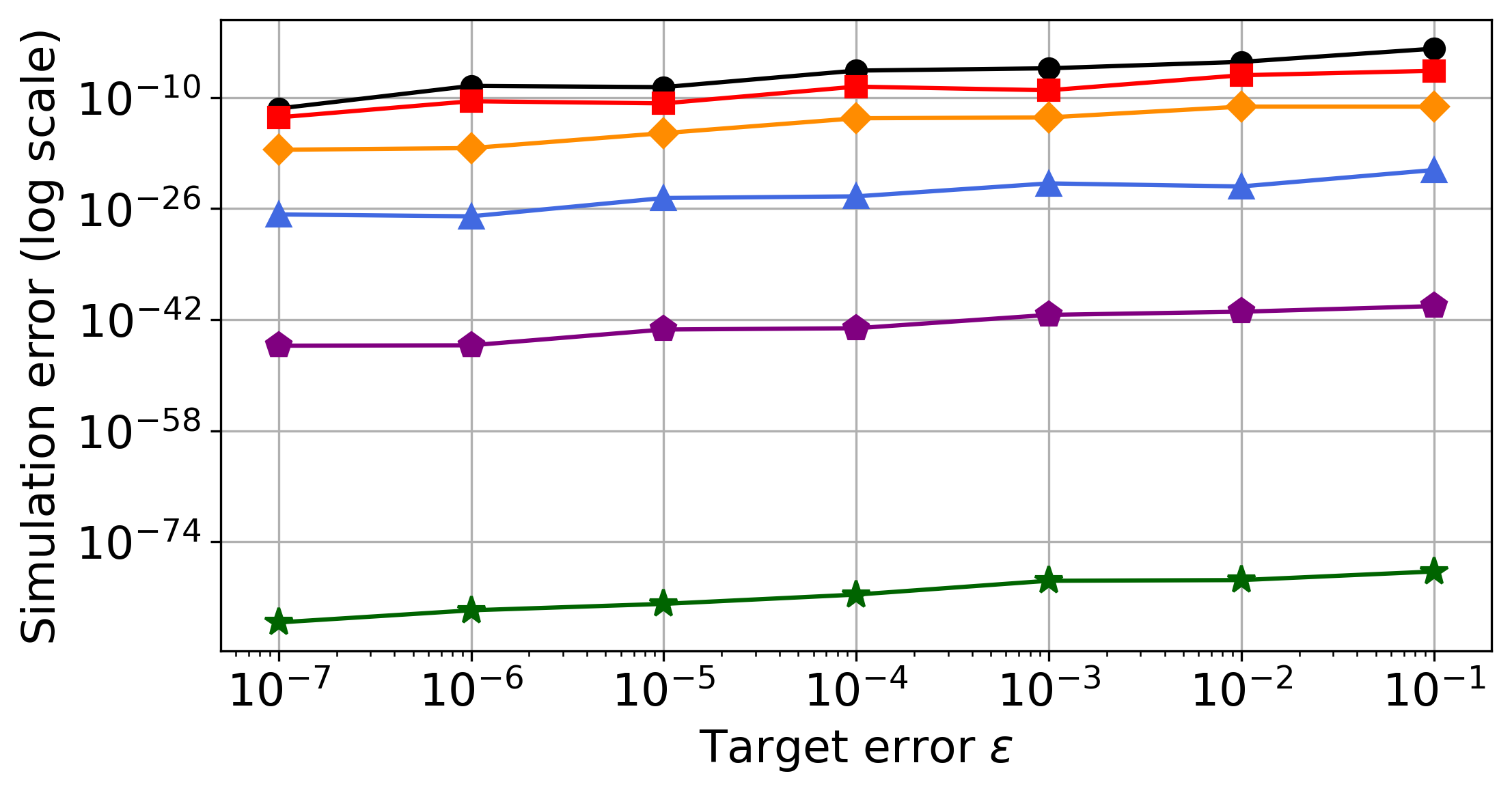}
        \label{fig:ln_geo_exp1}
    \end{subfigure}
    \hfill
    \begin{subfigure}{0.49\linewidth}
        \centering
        \caption{Arithmetically decaying}
        \includegraphics[width=\linewidth]{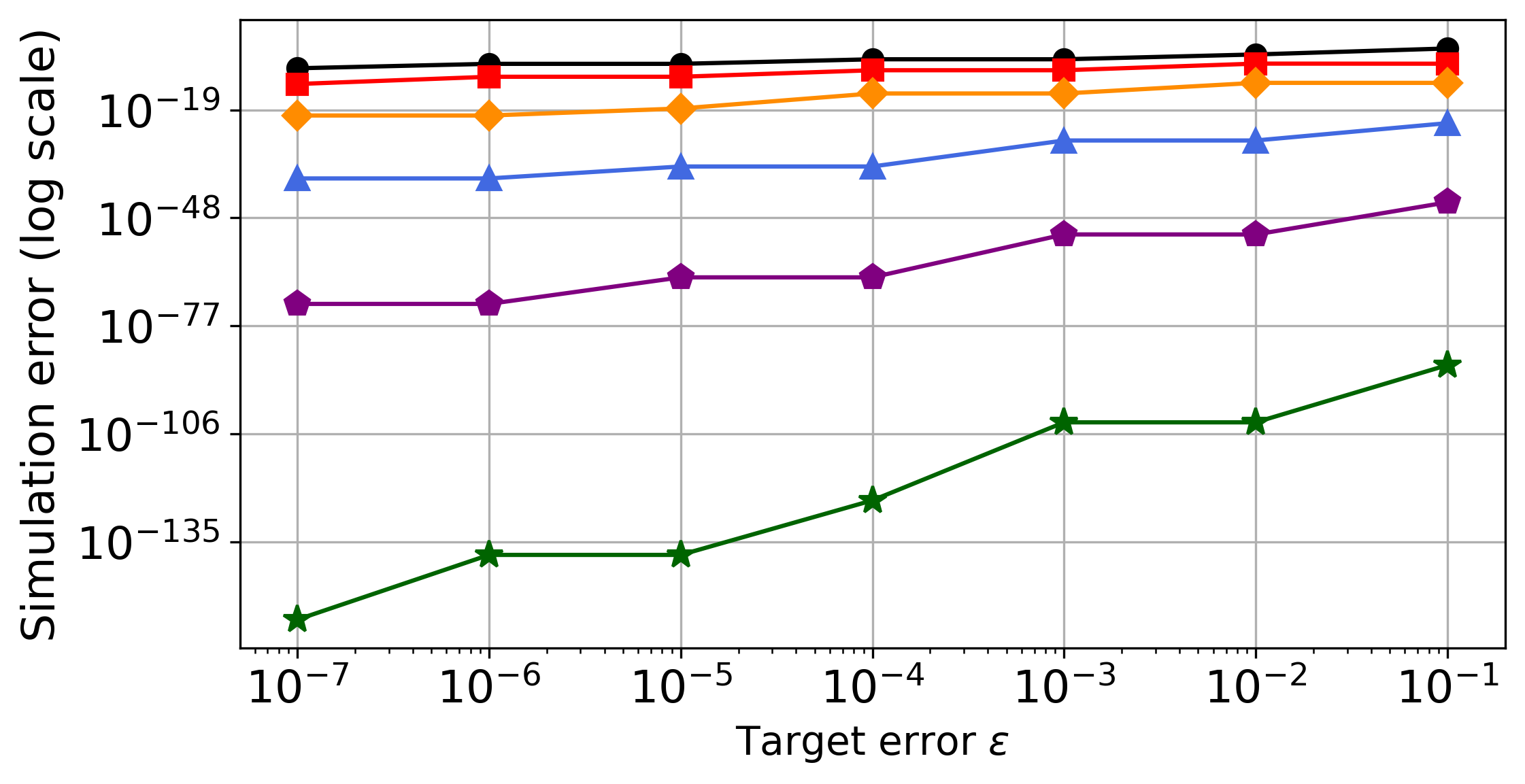}
        \label{fig:ln_arith_exp1}
    \end{subfigure}

    \vspace{1em} % 행 간 간격

    % 두 번째 행
    \begin{subfigure}{0.49\linewidth}
        \centering
        \caption{One dominant}
        \includegraphics[width=\linewidth]{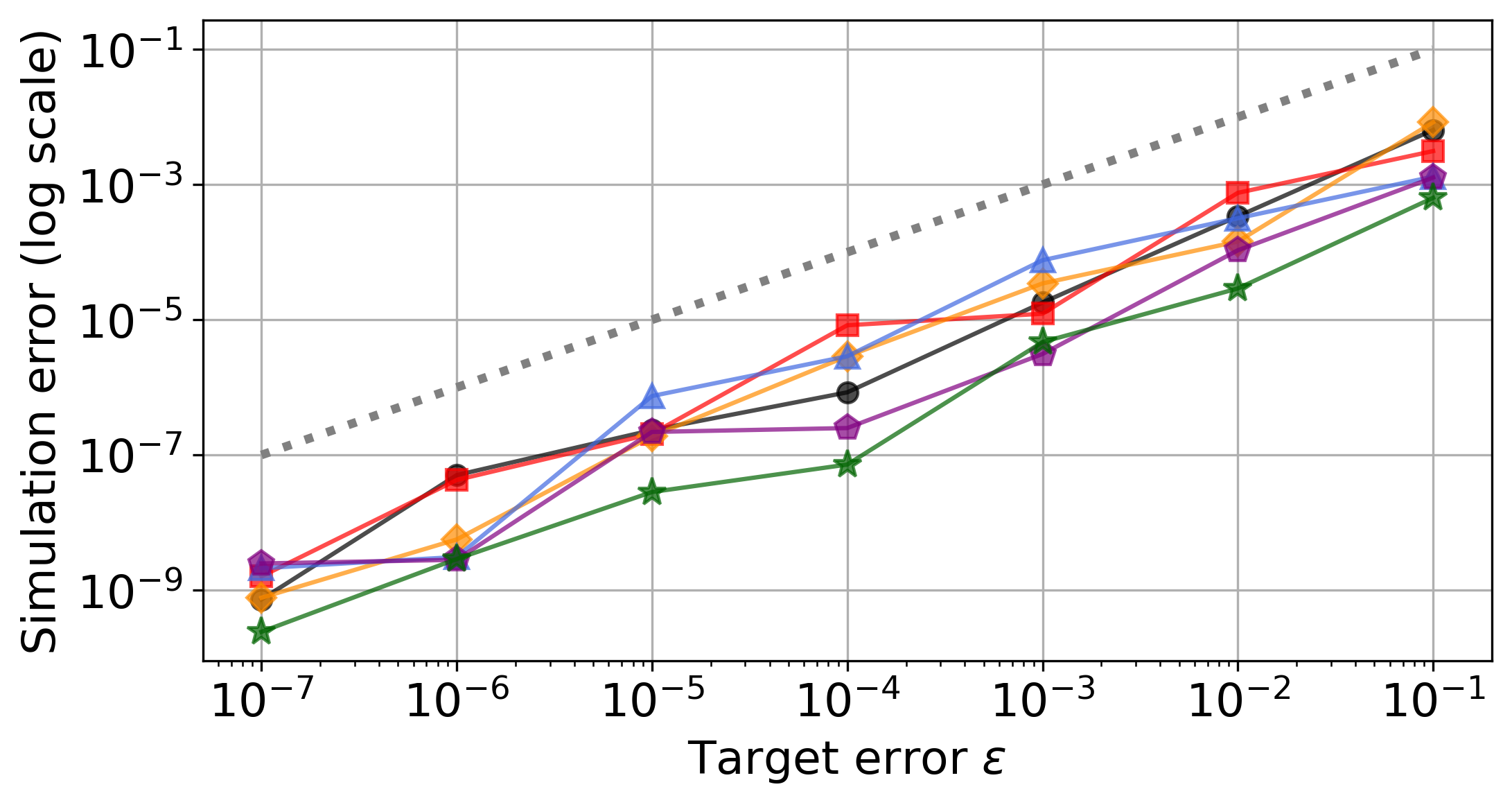}
        \label{fig:ln_dominant_exp1}
    \end{subfigure}
    \hfill
    \begin{subfigure}{0.49\linewidth}
        \centering
        \caption{Identical}
        \includegraphics[width=\linewidth]{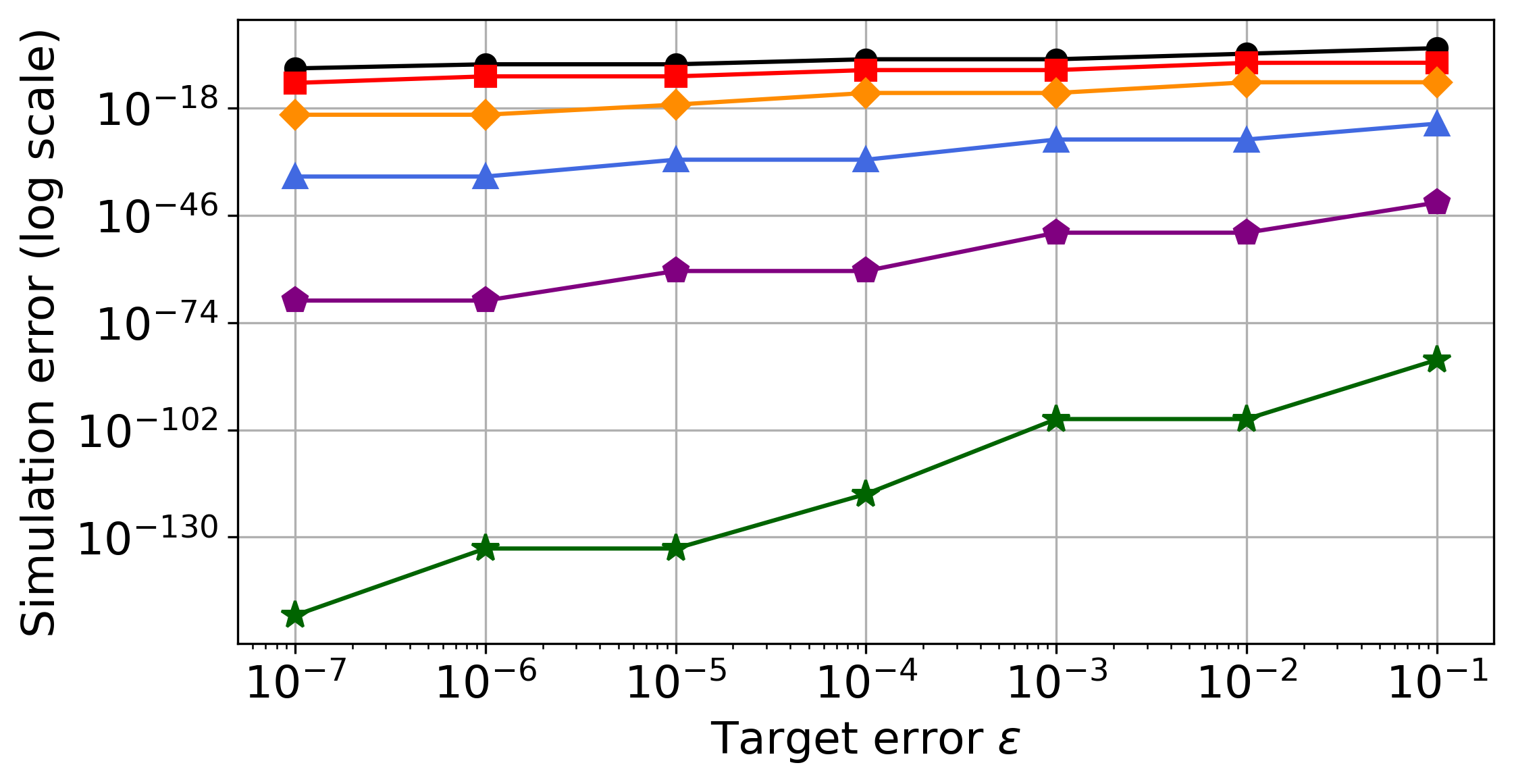}
        \label{fig:ln_id_exp1}
    \end{subfigure}

    \caption{Simulation results obtained by modifying $t$ according to~\cref{eq:new-bound} in the Scenario 1 described in~\cref{chap:simul-setup}.}
    \label{fig:ln-new-sim}
\end{figure}

\end{document}